%% file: main.tex
\documentclass[12pt,psfig]{article}

\usepackage[margin=1in]{geometry} 
\usepackage{setspace}              
\usepackage{url}
\usepackage{type1cm,verbatim,natbib}
\usepackage{graphicx}
\usepackage{latexsym}
\usepackage{lscape}
\usepackage{setspace,color,bbm,amsfonts}
\usepackage{amsthm}
\usepackage{float} 
\usepackage{bm}
\usepackage[dvipsnames]{xcolor}
\usepackage{longtable}

\renewcommand\thetable{\arabic{table}}

\usepackage{epsfig}
\usepackage{amsmath}

\def\bl{{\bf l}}

\def\bl{\begin{flushleft}}
\def\el{\end{flushleft}}
\newcommand{\bcr}{\begin{center}}
\newcommand{\ecr}{\end{center}}
\newcommand{\blind}{0}

\fontsize{40pt}{\baselineskip}\selectfont

\newtheorem{theorem}{Theorem}

\begin{document}           
\if0\blind
{
  \title{\bf On Cluster Randomized Trials with the Desirability of Outcome Ranking (DOOR) Endpoints}
  \author{
    Wanying Shao$^1$, 
    Toshimitsu Hamasaki$^1$, Scott Evans$^1$ \& Guoqing Diao$^1$\footnote{Corresponding Author: Guoqing Diao, George Washington University; Email:gdiao@gwu.edu.} \\
    $^1$    Department of Biostatistics and Bioinformatics \\
    The George Washington University, 
    Washington, D.C.\\
  }
\date{}
  \maketitle
} \fi

\if1\blind
{
  \bigskip
  \bigskip
  \bigskip
  \begin{center}
    {\LARGE\bf On Cluster Randomized Trials with the Desirability of Outcome Ranking (DOOR) Endpoints}
\end{center}
  \medskip
} \fi

\vspace{-10mm}
\begin{abstract}
\textcolor{black}{Cluster randomized trials are widely used when individual randomization is logistically
infeasible or when correlations between observations cannot be ignored, especially in fields such as
ophthalmology, infectious disease, vaccine research, and sociology. The desirability of outcome ranking
(DOOR) framework evaluates patient-centric benefit–risk using an ordinal outcome and a Wilcoxon-Mann-Whitney
statistic-based approach to compare outcome distributions between interventions. We propose a suite of new methods to extend DOOR to cluster trials based on
properties of U-statistics and influence functions to estimate within-cluster and between-cluster treatment
effects. These approaches can be applied in different scenarios, including mixtures of clusters with two
treatment groups and clusters with only one group, and both small and large numbers of clusters. Simulations
demonstrate that the proposed methods perform well under various scenarios regarding the number of clusters
and cluster sizes. As an illustration, we apply the proposed methods to a cluster randomized crossover trial comparing delayed cord clamping and umbilical cord milking for newborns. }
\end{abstract}
\vspace{-5mm}
{\it Keywords:} Within-cluster correlation; Influence function; Wilcoxon-Mann-Whitney test; U-statistics

\vspace{-5mm}
\section{Introduction}
Cluster randomized trials (CRT) are widely used in ophthalmology, infectious diseases,  vaccine research, and sociological studies. When considering interventions against infectious diseases, it is essential to account for indirect effects, in which individuals benefit from interventions provided to other members of the community \citep{hayesClusterRandomisedTrials2017}.  In sociology, units such as schools, communities, or families are often randomized as intact groups \citep{ cookEmergentPrinciplesDesign2005,spybrookProgressDecadeExamination2016}. Cluster randomized trials are therefore preferred when indirect effects need to be measured, correlations between observations cannot be ignored, or individual randomization is logistically infeasible. The need to avoid treatment contamination is also a common reason for choosing CRTs \citep{donnerAspectsDesignAnalysis1998,cookStatisticalLessonsLearned2016}. 
In \textcolor{black}{ophthalmology} research, clinical trials often involve clustered data when both eyes of a patient are enrolled in the study \citep{rosnerStatisticalMethodsOphthalmology1982}.

\noindent
In CRTs, observations within the same cluster are more alike than observations across clusters. Failing to account for the intracluster correlation coefficient (ICC) will underestimate the variance of the parameter estimator for the intervention effect and lead to misleading results \citep{jungClusterRandomizationTrials2024}. Due to positive ICCs, CRTs usually require a larger total sample size than individually randomized trials \citep{hemmingKeyConsiderationsDesigning2023}. Furthermore, the statistical power of a CRT is driven more by the number of clusters than by the cluster sizes \citep{wangChoosingUnitRandomization2024}.
\\\\
A large amount of literature has taken into account the correlation among units within a cluster and has proposed an adjusted chi-square test \citep{donnerAnalysisSitespecificData1988,donnerStatisticalMethodsOphthalmology1989,raoSimpleMethodAnalysis1992,ahnEvaluationWeightedChiSquare2003} and sample size calculations \citep{jungSampleSizeCalculations2001,kangSampleSizeCalculation2003} for binary outcomes. Continuous outcomes are commonly analyzed using mixed-effects models \citep{donnerRegressionApproachAnalysis, lairdMaximumLikelihoodComputations1987, hedekerRandomeffectsRegressionModels1994} or generalized estimating equations (GEE) \citep{liangLongitudinalDataAnalysis1986}. Some work has focused on clustered ordinal outcomes using parametric methods \citep{rosnerMultivariateMethodsClustered1997} and then considered incorporating the clustering effects for non-parametric methods, such as the Mann-Whitney U test \citep{rosnerUseMannWhitney1999} and the Wilcoxon signed-rank test \citep{rosnerWilcoxonSignedRank2006,dattaSignedRankTestClustered2008}. \citet{rosnerIncorporationClusteringEffects2003a} extended the Wilcoxon rank-sum test under the assumptions that all subunit observations within a cluster belong to the same treatment group, that observations within any cluster are exchangeable. \citet{rosnerExtensionRankSum2006} extended their approach to accommodate the situation where observations of a cluster can be assigned to different treatment groups, but still assumed exchangeability with the same intra-cluster dependence across groups. As a continuation of this work, \citet{rosnerPowerSampleSize2011} developed corresponding power and sample size calculations. The assumptions of \citet{rosnerIncorporationClusteringEffects2003a} were relaxed in the approach of \citet{dattaRankSumTestsClustered2005}, which is based on within-cluster resampling and remains valid when the cluster sizes are informative (i.e., when the outcomes and the cluster size are correlated). \citet{duttaRanksumTestClustered2016} extended the idea of within-cluster resampling to further accommodate the case where the number of observations in a group within a cluster is informative. Nevertheless, the \citet{dattaRankSumTestsClustered2005} method cannot be applied to very small cluster sizes, such as strictly contralateral designs. \citet{larocqueTwoSampleTests2010b} developed a Mann-Whitney-type test 
when clusters can contain observations from both groups, with no other assumptions made about the scales or shapes of the two distributions. However, they focused on the between-cluster comparison while with-cluster differences between treatments are ignored. 
\\\\
Motivated by studies in infectious diseases, \citet{evansDesirabilityOutcomeRanking2015a} proposed the Desirability of Outcome Ranking (DOOR) to evaluate benefits and harms at the patient level as an ordinal outcome using the Wilcoxon-Mann-Whitney test while accounting for ties. Patients were ranked according to a predefined ordinal overall clinical outcome. The DOOR framework integrates the efficacy and safety results that are possibly associated; in addition, it acknowledges the cumulative nature of patient outcomes and appropriately addresses competing-risk complexities.
\citet{hamasakiPatientcentricParadigmTool2025} described its applications to the analysis of clinical trials, and \citet{hamasakiDesirabilityOutcomeRanking2025} outlined the motivations for using DOOR as an analytic framework. \citet{shuLongitudinalBenefitRisk2025} extended DOOR to longitudinal studies, handling correlations between longitudinal outcomes. The DOOR endpoint has been widely used in clinical trials and observational studies of infectious, neurologic, and obstetrical diseases \citep{tammaClinicalImpactCeftriaxone2022, chamberlainDesirabilityOutcomeRanking2023, turnerDalbavancinTreatmentStaphylococcus2025,sandovalDesirabilityOutcomeRanking2024}. 
\\\\
A recent work by \citet{fangSampleSizeDetermination2025} developed a unified framework for  sample size and power calculations for win statistics (win ratio, win odds, and net benefit) in parallel cluster randomized trials, accounting for the rank intra-class correlation coefficient, cluster size variability, tie probability, and outcome prioritization. The DOOR probability can be expressed as a linear transformation of the net benefit \citep{buyseGeneralizedPairwiseComparisons2010}, with the two measures differing in how ``Win'' is determined. However, additional considerations and designs arise in ophthalmology studies. Each eye of a patient may receive different treatments, which requires attention to the correlation between the eyes \citep{rosnerStatisticalMethodsOphthalmology1982}. For example, in the study described by  \cite{pall2019management}, there were three groups:(1) contralateral/split-eye group, in which subjects received the test lens in one eye and the control lens in the contralateral eye; (2) test lenses group in which both eyes received the same intervention; and (3) control lenses group, in which both eyes received the control lens. Other recent examples in ophthalmology research include \cite{lyu2020comparison} and \cite{price2021randomized}. 
\\\\
Another motivating example is the Milking In Non-vigorous Infants (MINVI) study,  which was a cluster randomized crossover trial conducted between January 2019 and May 2021, and primary outcomes were available for 1,730 newborns from 10 medical centers \citep{katheriaUmbilicalCordMilking2023}. The study hypothesized that umbilical cord milking (UCM) would reduce admission to the neonatal intensive care unit (NICU) compared to early cord clamping (ECC).  Medical centers were randomized 1:1 to UCM or ECC in Period One and then crossed over to the other intervention during Period Two, with infants as the unit of analysis. The primary outcome was NICU admission within the first 24 hours of life based on predefined clinical criteria, while the principal safety outcome was hypoxic-ischemic encephalopathy (HIE). Clinicians were also interested in capturing the complete story of the patient experience using a DOOR framework. To our knowledge, currently there are no methods available for all types of cluster randomized trials with DOOR endpoints. 
\\\\
In this article, we develop several new approaches to measure both the within-cluster and between-cluster effects. The proposed methods accommodate different scenarios in practical cluster randomized trials, including small cluster sizes and a small number of clusters. Furthermore, the proposed methods are applicable to scenarios in which subjects in some clusters receive the same treatment, whereas subjects in other clusters are randomly assigned to two treatment arms. This type of design is particularly common in ophthalmology trials \citep{pall2019management, lyu2020comparison, price2021randomized}. In Section 2, we define the within- and between-cluster DOOR probabilities, and a mixture of these two DOOR probabilities. We develop estimation and inference procedures using the properties of U-statistics and influence functions. We then propose a novel test statistic in Section 3 to test the within- and between-cluster treatment effects while controlling for the family-wise error rate and describe how to assess the between-cluster variability. We conduct extensive simulation studies to examine the performance of the proposed methodology in Section 4 and provide a real application in Section 5. We conclude the paper with some discussions in Section 6. Technical details are deferred to the Online Supplement. 

\vspace{-7mm}
\section{Methods}
\vspace{-6mm}
Before presenting the proposed methods, we 
list the frequently used notation in the following table for ease of presentation. The unit of study subject can be a patient in most studies or an eye of a patient in ophthalmology studies. 
\begin{longtable}{|c|l|}
\hline
$n$ & Number of clusters  \\ \hline
$m_{i}$ & Number of subjects in cluster $i$  \\ \hline
$m_{i1}$ & Number of subjects in cluster $i$ assigned to treatment group  \\ \hline
$m_{i2}$ & Number of subjects in cluster $i$ assigned to  control group \\ \hline
$N_1$ & Total number of subjects assigned to treatment group         \\ \hline
$N_2$ & Total number of subjects assigned to control group         \\ \hline
$K$ & Number of DOOR levels        \\ \hline
$Y_{ij}$ & DOOR rank of the $j$th subject in the $i$th cluster         \\ \hline
$A_{ij}$ & Treatment group indicator of the $j$th subject in the $i$th cluster         \\ \hline
$I(\cdot)$ & Indicator function \\ \hline
$D_{wi}$ & DOOR probability for cluster $i$         \\ \hline
$D_{w}$ & Within-cluster DOOR probability    \\ \hline 
$D_{b}$ & Between-cluster DOOR probability        \\ \hline
$\widehat{D}_{wi}$ & DOOR probability for cluster $i$ estimator      \\ \hline
$\widehat{D}_w$ & Inverse variance weighted within-cluster DOOR probability estimator       \\ \hline
$\widehat{w}_i$ & Weight for the $i$th cluster used in $\widehat{D}_w$     \\ \hline
$\widetilde{D}_w$ & Sample size weighted within-cluster DOOR probability estimator       \\ \hline
$\widetilde{w}_i$ & Weight for the $i$th cluster used in $\widetilde{D}_w$     \\ \hline
$\widehat{D}_b$ & Between-cluster DOOR probability estimator      \\ \hline
$\widehat{W}_{w}$ & Test statistic for testing $H_0:D_w=0$ based on $\widehat{D}_w$    \\ \hline
$\widetilde{W}_{w}$ & Test statistic for testing $H_0:D_w=0$ based on $\widetilde{D}_w$    \\ \hline
$\widehat{W}_{b}$ & Test statistic for testing $H_0:D_b=0$ based on $\widehat{D}_b$    \\ \hline
$\widehat{W}_{v}$ & Test statistic for testing $H_0:D_b-D_w=0$ based on $\widehat{D}_b$ and $\widehat{D}_w$    \\ \hline
$\widetilde{W}_{v}$ & Test statistic for testing $H_0:D_b-D_w=0$ based on $\widehat{D}_b$ and $\widetilde{D}_w$    \\ \hline
$\widehat{W}_{\max}$ & Test statistic for testing $H_0:D_b=D_w=0.5$ based on $\widehat{D}_b$ and $\widehat{D}_w$    \\ \hline
$\widetilde{W}_{\max}$ & Test statistic for testing $H_0:D_b=D_w=0.5$ based on $\widehat{D}_b$ and $\widetilde{D}_w$      \\ \hline
$\widehat{D}_{wt}$ & Weighted average of $\widehat{D}_{w}$ and $\widehat{D}_{b}$     \\ \hline
$\widetilde{D}_{wt}$ & Weighted average of $\widetilde{D}_{w}$ and $\widehat{D}_{b}$   \\ \hline
$\widehat{W}_{wt}$ & Test statistic for testing $H_0:D_b=D_w=0.5$ based on $\widehat{D}_{wt}$    \\ \hline
$\widetilde{W}_{wt}$ & Test statistic for testing $H_0:D_b=D_w=0.5$ based on $\widetilde{D}_{wt}$      \\ \hline
$\sigma_i^2$ & Asymptotic variance of $\widehat{D}_{wi}$  \\ \hline
$\psi_{wi}$ & Influence function of $D_w$ for the $i$th cluster    \\ \hline
$\psi_{bi}$ & Influence function of $D_b$ for the $i$th cluster   \\ \hline
$r$ & $r=\lim_{m_i\rightarrow \infty} \frac{m_{i1}}{m_i}$   \\ \hline
$\rho$ & $\rho = \text{Cov}(\psi_{bi}, \psi_{wi})/\sqrt{\text{Var}(\psi_{bi})\text{Var}(\psi_{wi})}$ \\ \hline
$\rho_c$ & Within-cluster correlation \\
\hline
\end{longtable}

\subsection{Within-Cluster DOOR Probability}
\vspace{-6mm}
We first introduce the cluster-specific DOOR probability and describe how to estimate and make inferences on it for each cluster. Suppose that there are $n$ clusters with $m_i$ subjects in the $i$th cluster. For the $j$th subject in the $i$th cluster, let $Y_{ij}$ denote the DOOR rank and $A_{ij}$ the treatment indicator taking value 1 for the treatment group and 0 for the control group. Without loss of generality, we assume that the DOOR endpoint has $K$ levels, and a lower rank is preferred, whereas a higher rank is considered a less desirable outcome. For the $i$th cluster, the DOOR probability is defined as
\begin{align*}
  D_{wi} &= E\{I(Y_{ij} < Y_{ij'}) + I(Y_{ij}=Y_{ij'})/2 \},
\end{align*}
for a randomly selected pair such that $A_{ij}=1$ and $A_{ij'}=0$. Note that in parallel cluster randomized trials, all subjects in a cluster receive the same treatment, i.e., all $A_{ij}$s are  0 or 1. In this case, the within-cluster DOOR probability is not estimable. \cite{evansDesirabilityOutcomeRanking2015a} proposed to use a Wilcoxon-Mann-Whitney statistic corrected for ties to estimate the DOOR probability for a two-sample trial design. Thus, we estimate $D_{wi}$ by 
\begin{align*}
\widehat{D}_{wi}&=\frac{1}{m_{i1}m_{i2}}\sum_{j=1}^{m_{i}}\sum_{j'=1}^{m_{i}}\phi(Y_{ij},Y_{ij'}),
\end{align*}
where $\phi(Y_{ij},Y_{ij'})=A_{ij}(1-A_{ij'})\{I(Y_{ij} < Y_{ij'}) + I(Y_{ij}=Y_{ij'})/2\}$, and {$m_{i1}=\sum_{j=1}^{m_i} A_{ij}$ and $m_{i2} = m_i - m_{i1}$ are the numbers of subjects in the treatment group and control group, respectively. By Theorem 10.13 in \citet{hunterNotesGraduatelevelCourse}, we can show that
$\sqrt{m_i}(\widehat{D}_{wi}-D_{wi})\xrightarrow{d}N(0,\sigma_i^2)$ when $m_i\to \infty$, 
where $\sigma_i^2=\frac{\sigma_{10i}^2}{r}+\frac{\sigma_{01i}^2}{1-r}$, $r=\lim_{m_i\rightarrow \infty} \frac{m_{i1}}{m_i}$, $\sigma_{10i}^2=\text{Cov}(\phi(Y_{ij},Y_{il}),\phi(Y_{ij},Y_{il'}))$, and $\sigma_{01i}^2=\text{Cov}(\phi(Y_{ij},Y_{il}),\phi(Y_{ij'},Y_{il}))$ with   $A_{ij}=A_{ij'}=1$ and $A_{il}=A_{il'}=0$. Detailed expressions of $\sigma_{10i}^2$ and $\sigma_{01i}^2$ are provided in Supplement A. 
\\\\
Assuming a common cluster-specific DOOR probability across clusters (i.e., $D_{w1}=\cdots=D_{wn}$), referred to as the within-cluster DOOR probability $D_w$, 
we propose the following two approaches to combine the results from multiple clusters.

\vspace{-6mm}
\subsubsection{Inverse variance weighted within-cluster DOOR probability estimator}
\vspace{-6mm}
The first approach  is to summarize the DOOR probabilities from the $n$ clusters using optimal weights. In particular, we estimate $D_w$ by $\widehat{D}_{w} = \sum_{i=1}^n \widehat{w}_i \widehat{D}_{wi}$, where $\widehat{w}_i = m_i\widehat{\sigma}_i^{-2}/\sum_{i'=1}^n m_{i'} \widehat{\sigma}_{i'}^{-2}$ and $\widehat{\sigma}_i^2$ is a consistent estimator of $\sigma_i^2$.  It is straightforward that $\text{Var}(\widehat{D}_{w})$ can be consistently estimated by $1/{\sum_{i=1}^n m_i \widehat{\sigma}_i^{-2}}$.
We note that this weighting approach is the same as the fixed effects meta-analysis approach. We define $\psi_{wi} = nm_i \sigma_i^{-2}(\widehat{D}_{wi} - D_w)/\sum_{i'=1}^n m_{i'} \sigma_{i'}^{-2}$, which can be viewed as the influence function of $D_w$ for the $i$th cluster.

\vspace{-6mm}
\subsubsection{Sample size weighted within-cluster DOOR probability estimator}
\vspace{-6mm}

 The optimal weighting approach described above may fail when the number of observations in a cluster is small or the correlation between $Y_{ij}$ and $Y_{ij'}$ is high, leading to an unstable and near zero estimate of  the variance of $\widehat{D}_{wi}$. Alternatively, we assign weights based on the sample size within each cluster and treatment group and obtain the following estimator 
\begin{align*}
\widetilde{D}_{w} = & \frac{1}{\sum_{g=1}^n m_{g1} m_{g2}}\sum_{i=1}^{n} \sum_{j=1}^{m_{i}}\sum_{j'=1}^{m_{i}}\phi(Y_{ij},Y_{ij'})
= \sum_{i=1}^n \widetilde{w}_i \widehat{D}_{wi}
\end{align*}
where $\widetilde{w}_i = m_{i1}m_{i2} / \sum_{g=1}^n m_{g1} m_{g2}$. 
 The influence function of $D_w$ for the $i$th cluster, with a slight abuse of notation, is $\psi_{wi}=n\widetilde{w}_i(\widehat{D}_{wi} - D_w)$. When the number of clusters is large, the variance of $\widetilde{D}_{W}$ can be estimated using the empirical version
$\sum_{i=1}^n \widetilde{w}_i^2 (\widehat{D}_{wi} - \widetilde{D}_{w})^2$ (Type 1 method).
Alternatively, if the number of clusters is small but the cluster size is large, we can estimate 
$\text{Var}(\widetilde{D}_{w})$ with $\sum_{i=1}^n \widetilde{w}_i^2\sigma_i^2$ (Type 2 method).

\subsubsection{Small-sample correction}
\vspace{-6mm}
In practice, the number of clusters is often small in a CRT and hence the large sample approximation for the standard error estimation of $\widetilde{D}_w$ may not be accurate. To correct for small-sample biases in standard error estimates, we can use the sample variance of the influence function and then multiply it by $1/(n-2)$ to estimate $\text{Var}(\widetilde{D}_w)$.  As noted by \citet{cornfieldSYMPOSIUMCHDPREVENTION1978}, ``one additional inflation factor requires consideration, the degrees of freedom with which $\sigma^2$ is estimated.'' Accordingly, critical values are taken from $t_{n-1}$ instead of N(0,1) to construct confidence intervals and perform hypothesis testing (Type 3 method). The same corrections can be applied to the estimator of the between-cluster DOOR probability described in the subsequent subsection.

\vspace{-6mm}
\subsection{Between-Cluster DOOR Probability}
\vspace{-6mm}
As described above, the within-cluster DOOR probability is not estimable when all subjects in a cluster receive the same treatment, for example, in a parallel cluster randomized trial. Instead, we consider the so-called between-cluster DOOR probability, defined as 
\[
D_b = E\{I(Y_{ij} < Y_{i'j'}) + I(Y_{ij}=Y_{i'j'})/2 \},
\]
for a randomly selected pair such that $A_{ij}=1$, $A_{i'j'}=0$, and $i\neq i'$. We consider a trial design in which subjects within a cluster could be assigned to either treatment or control, and estimate $D_b$ as $\widehat{D}_b=\frac{1}{N_1 N_2-\sum_{g=1}^n m_{g1}m_{g2}}\sum_{1\leq k < i \leq n}(\Phi_{ik}+\Phi_{ki})$, where $N_1$ and $N_2$ are the total number of observations from treatment and control groups, respectively, and $\Phi_{ik}=\sum_{j=1}^{m_{i1}}\sum_{l=1}^{m_{k2}}I(A_{ij}=1,A_{kl}=0)\phi(Y_{ij},Y_{kl})$, $k\neq i$. 
 \\\\
Let $\mathbb{P}_n$ be the empirical measure based on $n$ i.i.d. observations (clusters).  Define $h({\bm Y}_i, {\bm Y}_k)= \frac{\binom{n}{2}}{N_1 N_2-\sum_{g=1}^n m_{g1}m_{g2}}(\Phi_{ik} + \Phi_{ki})$, where ${\bm Y}_i$ and ${\bm Y}_k$ are the observations in the $i$th and $k$th clusters, respectively. Using the properties of a $U$-statistic, we can prove the following theorem.
\begin{theorem}
Let
$\psi_{bi} = 2 E\!\left[h(\widetilde{ \boldsymbol{Y}}_i, \boldsymbol{Y}_i)-D_b \mid \boldsymbol{Y}_i\right]$,
where $\widetilde{ \boldsymbol{Y}}_i$ is an independent copy of   $\boldsymbol{Y}_i$, then $\psi_{bi}$ is the influence function of $D_b$ and
$\sqrt{n}\,(\widehat{D}_b - D_b)
= \sqrt{n}\mathbb{P}_n\psi_b + o_p(1)$.
Consequently,
$\sqrt{n}\,(\widehat D_b - D_b)\ \xrightarrow{d}\ N\!\left(0,\ \text{Var}(\psi_{bi})\right)$ as $n\rightarrow \infty$.
Moreover, the variance of $\psi_{bi}$ can be consistently estimated by
$\sum_{i=1}^n \widehat{\psi}_{bi}^2/n$, where $\widehat{\psi}_{bi}$ is the empirical version of $\psi_{bi}$ with the unknown parameters replaced by their corresponding consistent estimators. 
\end{theorem}
\noindent
In Theorem 1, we derive the asymptotic variance for $\widehat{D}_b$ using influence functions. Detailed proofs are given in Supplement B. When the number of clusters is small, but $n>=6$, we use the small-sample correction technique described in subsection 2.1.2. Our estimator $\widehat{D}_b$ is similar to the Mann-Whitney statistic for comparisons between clusters  proposed by \citet{larocqueTwoSampleTests2010b}. However, we incorporate within- and between-cluster treatment effects described in the next section, whereas \citet{larocqueTwoSampleTests2010b} focused on the between-cluster treatment effect.

\vspace{-6mm}

\subsection{Within-Cluster versus Between-Cluster DOOR Probabilities}
\vspace{-6mm}

There is a clear distinction between the definitions of within- and between-cluster DOOR probabilities. The within-cluster DOOR probability captures the treatment effect within the same cluster and controls for observed or unobserved cluster-specific factors that are constant within the cluster. On the other hand, 
the between-cluster DOOR probability reflects how differences in the level of  treatment between clusters are related to differences in the DOOR endpoint between those same clusters. In a simple linear regression setting with continuous endpoints, the between-cluster treatment effect is essentially 
the slope of the cluster-level regression of outcomes on treatments \citep{neuhaus1998between, roberts2005design}. In the absence of between-cluster variability, all observations within and across clusters are mutually independent,  and hence these two DOOR probabilities are equivalent. However, when there is between-cluster variability, the between-cluster DOOR probability can be confounded by cluster-level characteristics that correlate with both treatment assignments and outcomes. The consequences of using the between-cluster treatment effect (or the between-cluster DOOR probability here) can be particularly severe when the between-cluster variability is large and the number of clusters is small. We demonstrate the potential fallacy of between-cluster DOOR probabilities using several examples, presented in Figure \ref{fig:between_clu_vari}. In the top-left panel, there is a clear treatment effect within each cluster but the outcomes in cluster 2 are systematically better than those in cluster 1, regardless of treatment assignment. Using $D_b$ will fail to detect the true treatment effect, leading to a false negative result. In the top-right panel, there is no treatment effect within each cluster. However, due to the imbalance in treatment assignment between these two clusters (more patients in the control arm in cluster 1 than cluster 2) and a large between-cluster variability such that outcomes in cluster 2 are systematically better than those in cluster 1, using $D_b$ will detect a spurious treatment effect, leading to a false positive result. The lower two panels present examples from standard parallel cluster randomized trials, where the gray dots indicate the potential outcomes of patients had they received the opposite treatment. Using $D_b$ will lead to a false negative result under the scenario in the bottom-left panel, whereas it will lead to a false positive result under the scenario in the bottom-right panel. These examples highlight that the between-cluster DOOR probability can be driven by cluster-level factors that are unrelated to the direct effect of treatment on an individual.

\section{Hypothesis Testing}
\subsection{Assess the between-cluster variability}
As a byproduct of the theoretical properties of  the estimators of the within- and between-cluster DOOR probabilities, $\widehat{D}_w$ and $\widehat{D}_b$, we can assess the between-cluster variability by testing $H_0: D_b-D_w=0$. A natural test statistic is  the following Wald-type statistic:
$\widehat{W}_v=(\widehat{D}_b - \widehat{D}_w)/\widehat{\text{S.E.}}(\widehat{D}_b - \widehat{D}_w)$. Based on the influence function representation, we can show that $\sqrt{n}(\widehat{D}_b - D_b,\widehat{D}_w-D_w)$ converges in distribution to a bivariate normal distribution with mean ${\bf 0}$ 
and variance-covariance matrix 
\[
{\boldsymbol{\Sigma}} = 
    \begin{pmatrix}
        \text{Var}(\psi_{bi}) & \text{Cov}(\psi_{bi}, \psi_{wi}) \\
        \text{Cov}(\psi_{wi}, \psi_{bi}) & \text{Var}(\psi_{wi})\\
    \end{pmatrix},
\]
where $\psi_{wi}$ and $\psi_{bi}$ are the previously defined influence functions for $D_w$ and $D_b$, respectively.  We replace the unknown parameters in $\boldsymbol{\Sigma}$ with their consistent estimators and obtain a consistent estimator of $\boldsymbol{\Sigma}$, denoted by $\widehat{\boldsymbol{\Sigma}}$. By linear contrast, we have $\widehat{\text{S.E.}}(\widehat{D}_b - \widehat{D}_w)
=\sqrt{\boldsymbol c^\top \widehat{\boldsymbol\Sigma} \boldsymbol c/n}$, where $\boldsymbol{c}=(1,-1)^\top$. Under the null hypothesis that there is no between-cluster variability, $W_v$ converges to a standard normal distribution as $n\rightarrow \infty$. The small-sample correction method described in the previous sections can also be applied here. In a similar fashion, we can define a Wald type statistic based on $(\widehat{D}_b, \widetilde{D}_w)$. We omit the details here.

\subsection{Simultaneous testing using $L_\infty$ norm}

When there is little or no between-cluster variability, neither the within-cluster DOOR probability alone nor the between-cluster DOOR probability alone can fully capture the treatment effect when observations are assigned to two treatment groups within clusters. It is ideal to incorporate all available information in a test. Specifically, we consider $H_0: D_b=D_w=0.5$. 
\\\\
If the null hypothesis of $D_b-D_w=0$ is not rejected, hybrid methods that incorporate all available information are preferred. To test $H_0: D_w=D_b=0.5$, 
we first define two  test statistics
based on the estimators $\widehat{D}_b$ and $\widehat{D}_w$: 
$\widehat{W}_b = (\widehat{D}_b - 0.5)/\widehat{\text{S.E.}}(\widehat{D}_b)$ and $\widehat{W}_w = (\widehat{D}_w - 0.5)/\widehat{\text{S.E.}}(\widehat{D}_w)$. Under the null hypothesis,  the vector $(\widehat{W}_b,\widehat{W}_w)$ asymptotically follows a bivariate normal distribution with mean $\boldsymbol{0}$ and  variance-covariance matrix $\boldsymbol{R}$ with diagonal elements 1 and off-diagonal elements $\rho = \text{Cov}(\psi_{bi}, \psi_{wi})/\sqrt{\text{Var}(\psi_{bi})\text{Var}(\psi_{wi})}$. 
To test $H_0: D_b=D_w=0.5$, we consider the test that controls for the familywise error rate (FWER) \citep{tukey1953problem}. Specifically, 
we define $\widehat{W}_{max}=\text{max}(\widehat{W}_w,\widehat{W}_b)$,  when $H_1$ is one-sided. A critical value $c$ adjusted for multiple tests should satisfy $P(\widehat{W}_{max}\le c \mid H_0)
= P(\widehat{W}_b \le c,\, \widehat{W}_w \le c |H_0)= 1 - \alpha$. A simultaneous confidence interval can also be constructed using the above critical value $c$. When $H_1$ is two-sided, we define $\widehat{W}_{max}=\text{max}(|\widehat{W}_w|,|\widehat{W}_b|)$, which can be considered the $L_\infty$ norm of $\widehat{W}_w$ and $\widehat{W}_b$, and the critical value $c$  satisfies $P(\widehat{W}_{max}\le c \mid H_0)= P(-c \le \widehat{W}_b \le c, -c \le \widehat{W}_w \le c \mid H_0)=1 - \alpha$. The probabilities here involve two-dimensional integrals that can be calculated using the {\it mvtnorm} R package. If the number of clusters is small, we use the small-sample correction described before and replace the bi-variate normal distribution with a bi-variate $t$-distribution in determining  $c$. 


\subsection{Weighted average of $D_b$ and $D_w$}

Alternative to the test controlling for FWER described above, we consider  
a new statistic based on the weighted average of $\widehat{D}_b$ and $\widehat{D}_w$, denoted by $\widehat{D}_{wt}= a_1\widehat{D}_w+a_2\widehat{D}_b$, where $a_1, a_2 \geq 0$ and $a_1+a_2=1$. 
Rather than fixing the values of the weights $a_1$ and $a_2$, we adopt the approach in fixed-effects multivariate meta-analysis using data-driven weights. The optimal weights that minimize the variance of $\widehat{D}_{wt}$ are $(a_1, a_2)^T = \frac{\boldsymbol{\Sigma}^{-1}\boldsymbol{j}}{\boldsymbol{j}^T\boldsymbol{\Sigma}^{-1}\boldsymbol{j}}$, where $\boldsymbol{j}$ is a vector of elements 1. The resulting variance of $\widehat{D}_{wt}$ is then approximately $\frac{1}{n\boldsymbol{j}^T\boldsymbol{\Sigma}^{-1}\boldsymbol{j}}$. With a slight abuse of the notation, we write $\widehat{D}_{wt} = \widehat{a}_1 \widehat{D}_w + \widehat{a}_2 \widehat{D}_b$ and consider the following test statistic $\widehat{W}_{wt} =\sqrt{n\boldsymbol{j}^T\widehat{\boldsymbol{\Sigma}}^{-1}\boldsymbol{j}}(\widehat{D}_{wt} - 0.5) $ for testing $H_0:D_b=D_w=0.5$.
Under the null hypothesis, $\widehat{W}_{wt}$ is asymptotically $N(0,1)$. Similarly, we can define $\widetilde{W}_{wt}$ based on $\widehat{D}_b$ and $\widetilde{D}_w$ and apply the small-sample correction.

\section{Simulation Study}

We conducted extensive simulation studies to examine the performance of the proposed estimation and inference methods. In the first set of simulations, we assess the estimation of the proposed within-cluster DOOR probability using optimal weights, within-cluster DOOR probability using weights proportional to the number of pairs in a cluster, and between-cluster DOOR probability under the alternative hypothesis. We generate the outcome data under a latent-variable random-effects model given by
\begin{align*}
    X_{ij}=\beta I(A_{ij}=1)+\alpha_i+\delta_{ij}\\
    i=1,2,...,n,\quad j=1,2,..., m_{i1}+m_{i2}
\end{align*}
where $\beta$ is the conditional treatment effect, $\alpha_i$ is a random intercept accounting for the common effect for cluster $i$, and $\delta_{ij}$ is an individual random error. We assume that  $\alpha_i$s and $\delta_{ij}$s are mutually independent from normal distributions $N(0,\sigma_{\alpha}^2)$ and $N(0,\sigma_{\delta}^2)$, respectively. 
Then, the DOOR rank of the $j$th participant in the $i$th cluster, $Y_{ij}$,  is defined based on $X_{ij}$ and pre-specified cut-points that lead to the target distribution of $Y_{ij}$. Five DOOR ranks are specified, with proportions 10\%, 20\%, 30\%, 25\%, and 15\% for ranks 1 to 5 in the control group, respectively. The cut-points are determined by the corresponding quantiles of the distribution of $X_{ij}$ in the control group, and the same cut-points are applied to both treatment and control groups. We set $\beta = 0.1$, $\sigma_{\alpha}^2=\rho_c/(1-\rho_c)$, and $\sigma_{\delta}^2=1$. It is clear that  the correlation between $X_{ij}$ and $X_{il}$ $(j\neq l)$ is equal to $\rho_c$. Then, we consider six different levels of within-cluster correlations: $\rho_c=0.001, 0.02, 0.06, 0.1, 0.3$, or $0.5$. With $\beta = 0.1$, the corresponding true values of $D_b$ are 0.5266, 0.5263, 0.5258, 0.5252, 0.5223, and 0.5188, and the true values of $D_w$ are 0.5266, 0.5266, 0.5265, 0.5264, 0.5257, and 0.5246.

\noindent
We consider two settings for the number of clusters. For the large number of cluster settings, we use 200, 100, and 50 clusters with corresponding cluster sizes of 4, 8, and 16, maintaining a total sample size of 800. For the small number of cluster settings, we use 10, 15, and 20 clusters with corresponding cluster sizes of 60, 40, and 30, maintaining a total sample size of 600. For each sample size setting, three treatment assignment scenarios are examined: (1) One-group randomization: all subjects within a cluster are assigned to the same group, with clusters randomized to either the treatment or control arm; (2) Two-group randomization: subjects within each cluster are independently randomized to the treatment or control group in a 1:1 ratio; (3) Mixture assignment: approximately 25\% of clusters are assigned entirely to the treatment group, another 25\% entirely to the control group, and the subjects in each cluster of the remaining 50\% of clusters are randomized to either treatment or placebo in a 1:1 ratio. All simulations are based on 10,000 replicates.
\\\\
 Table 1 presents the bias, the standard deviation of the estimates, the average of the standard error estimates, and the coverage probability of the 95\% confidence interval for the large-sample scenario with 100 clusters and 8 observations per cluster. Under large-sample scenarios, no corrections or adjustments are applied. In the one-group randomization scenario, only $D_b$ can be estimated. The results are shown for correlation coefficients of 0.01, 0.02, 0.06, 0.1, 0.3, and 0.5. Under the two-group and mixture randomization scenarios, $\widetilde{D}_w$ and $\widehat{D}_b$ are used to estimate $D_w$ and $D_b$, and results are reported for a correlation coefficient of 0.1. Additional results for other correlation coefficients and the remaining large sample scenario are provided in Supplement C.1.  Table 2 summarizes the results for the  small-sample scenario with 10 clusters and 60 observations per cluster. The estimators $\widehat{D}_w$, $\widetilde{D}_w$ with the Type 2 method, $\widetilde{D}_w$ with the Type 3 method, and $\widehat{D}_b$ with small-sample corrections are considered. 
The results in Tables 1 and 2 indicate that the proposed estimators, in general, perform well in both large and small-sample sizes, regardless of whether the assignment is one-group, two-group, or a mixture. In particular, 
 $\widehat{D}_b$ with small-sample corrections performs well in all randomization scenarios, and $\widetilde{D}_w$ with the Type 2 method for standard error estimation performs better than the Type 3 method in the mixture setting. 
\\\\
In the second set of simulations, we evaluate the Type I error rates and powers of the proposed tests for testing $H_0: D_b=0.5$, $H_0: D_w=0.5$, and $H_0: D_w=D_b=0.5$. We set $\beta=0$ when evaluating Type I error rates, and $\beta=0.25$ when evaluating powers in large sample scenarios, and $\beta=0.30$ in small-sample scenarios. All other settings are the same as the first set of simulations. The nominal significance level is set to 0.05 in all cases.  Figures \ref{fig:large_1_TIE} to \ref{fig:large_mix_TIE} present Type I error rates for large-sample scenarios. For one-group and two-group randomization, Type I error rates for all methods are generally between 0.045 and 0.06. For two-group randomization, when the correlation is larger than 0.001, the Type I error rate for testing $H_0:{D}_w=0.5$ is consistently higher than those based on other estimators. Type I error rates are slightly inflated under mixture randomization in some scenarios. Figures \ref{fig:large_1_power} to \ref{fig:large_mix_power} display powers for large-sample scenarios. For one-group randomization, as expected, the power for the test based on $\widehat{W}_b$ increases as the correlation decreases and also increases as the number of clusters increases. In practice, when the total sample size is fixed, we suggest enrolling more clusters rather than increasing cluster size. For two-group randomization, the power for the test based on $\widetilde{W}_{w}$ is slightly lower than that of the other statistics, while the power for the test based on $\widetilde{W}_{wt}$ is consistently the highest. For mixture randomization, only half of the clusters contribute to the within-cluster DOOR probability. In most panels, the power of the test based on $\widetilde{W}_{w}$ is lower than that of the other methods. The power of the test based on $\widetilde{W}_{max}$ lies between that of $\widetilde{W}_{w}$ and $\widehat{W}_b$, while the test based on $\widetilde{W}_{wt}$ has the highest power. 

\noindent
The results on the Type I error rates and powers for the small-sample scenarios are provided in Supplement C.2. For one-group randomization, Type I error rates for testing $H_0:{D}_b=0.5$ without small-sample corrections  are inflated, particularly when there is a within-cluster correlation and the number of clusters is 10. For two-group randomization and mixture randomization, Type I error rates for testing $H_0: {D}_{w}=0.5$ without small-sample corrections are also inflated; consequently, tests based on $\widetilde{W}_{max}$ and $\widetilde{W}_{wt}$ also have inflated Type I error rates. In contrast, most methods with small-sample corrections control the Type I error rate well at the nominal level, except for the test based on $\widehat{W}_b$ under one-group randomization when the correlation is large. Since the number of clusters is small, as expected, the power under the mixture setting is lower than that under the two-group setting. 
\\\\
 When cluster sizes are small, estimates of $D_w$ by inverse variance weights described in Section 2.1.1, as well as the corresponding test statistics $\widehat{W}_{max}$ and $\widehat{W}_{wt}$, are not available. In such cases, the variances of standard DOOR probabilities for some clusters are zero, precluding the use of inverse-variance weighting. In scenarios with a small number of clusters, the variances of $\widetilde{D}_{w}$ approximated using the influence function (Type 1 method), as well as that of $\widehat{D}_b$, tend to perform poorly, with coverage probabilities below the nominal 95\% level and inflated Type I error rates. This also affects the performance of tests based on $\widetilde{W}_{max}$ and $\widetilde{W}_{wt}$. In this setting, estimating the variance of $\widetilde{D}_{w}$ using Type 2 method, as well as estimating the variances of $\widehat{D}_b$ and $\widetilde{D}_{w}$ with the small-sample correction based on influence functions described in Section 2.1.3 (Type 3 method), is feasible. As expected, the test based on $\widetilde{W}_{max}$ captures both within-cluster and between-cluster information and generally performs reliably, with performance typically between that of tests based on $\widetilde{W}_{w}$ and $\widehat{W}_{b}$ in terms of Type I error and power. Moreover, the test based on $\widetilde{W}_{wt}$ is recommended, as it preserves both within-cluster and between-cluster information and usually achieves the highest power, particularly under mixture randomization designs. Guidance for the choice of the proposed method in different practical settings is provided in Table 3.


\section{Application}

The MINVI study is a cluster randomized crossover trial conducted among 1730 newborns from 10 medical centers between January 2019 and May 2021 \citep{katheriaUmbilicalCordMilking2023}. The study hypothesized that umbilical cord milking (UCM) would reduce admission to the neonatal intensive care unit (NICU) compared with early cord clamping (ECC). Hospitals were randomized in a 1:1 ratio to UCM or ECC in period one until half of the required enrollment was reached, then crossed over to the other intervention in period two for the remaining participants. The DOOR endpoint was one of the exploratory endpoints, and it was reported in six ordinal levels: 1) Alive without any of the following events: hypoxic-ischemic encephalopathy (HIE), NICU admission, or cardiorespiratory support in the delivery room, 2) Alive with cardiorespiratory support in the delivery room, 3) Alive with NICU admission for predefined criteria, 4) Alive with mild HIE, 5) Alive with moderate to severe HIE, and 6) Death. In the initial DOOR analysis, cluster
assignment was ignored, although a sensitivity analysis was conducted using a stratified DOOR analysis for each cluster and then
combined using a meta-analysis-derived DOOR probability \citep{katheriaApplicationDesirabilityOutcome2024}, equivalent to the inverse variance weighted within-cluster DOOR probability estimator defined in Section 2.1.1.
\\\\
To illustrate the use of the proposed methodology, we reanalyzed the data using the study-defined DOOR endpoint. The data for the analysis consist of 1,524 newborns with complete measurements on all DOOR components. The number of newborns at each site ranged from 34 to 289. In this case, given the small number of clusters, Figure \ref{fig:MINVI_study} summarizes the estimated DOOR probabilities with 95\% confidence intervals using $\widehat{D}_{w}$, $\widetilde{D}_{w}$ with Type 2 method, and $\widehat{D}_b$ with corrections, as suggested in Table 3. All estimated DOOR probabilities suggest that UCM is more desirable than ECC under the study-defined DOOR endpoint. To assess the between-cluster variability, we tested $H_0: D_b-D_w=0$. Using estimators $\widehat{D}_w$ and $\widehat{D}_b$, the test statistics $\widehat{W}_v$ was 0.26, whereas using estimators $\widetilde{D}_w$ and $\widehat{D}_b$, the test statistics $\widetilde{W}_v$ was 1.62. We failed to reject the null hypothesis in both cases. Therefore, hybrid methods that incorporate all available information are recommended. Subsequently, to test $H_0: D_w=D_b=0.5$, simultaneous 95\% CIs are (0.521, 0.592) for $\widehat{D}_{w}$ and (0.499, 0.627) for $\widehat{D}_b$ with small-sample corrections. While the simultaneous 95\% CIs are (0.492, 0.568) for $\widetilde{D}_{w}$ with Type 2 method and (0.500, 0.626) for $\widehat{D}_b$ with small-sample corrections. A bi-variate $t$-distribution was used to determine the critical value in both cases. The results based on $\widehat{W}_{max}$ and all other results presented in Figure \ref{fig:MINVI_study} are consistent, suggesting that UCM is more desirable. Considering the crossover design, the within-cluster DOOR probabilities using inverse-variance weights were additionally calculated separately for sequence 1 (conducted UCM in period 1, then ECC in period 2) and sequence 2 (conducted ECC in period 1, then UCM in period 2) and are shown in the same figure. The results based on $\widehat{D}_w$ for sequences 1 and 2 are consistent with those of all other methods. UCM performs better regardless of whether it was performed in the first or second period. 

\section{Discussion}

In this article, we have developed methodologies for cluster randomized trials with the DOOR endpoints, taking into account different designs. We have introduced two estimands: the within-cluster and the between-cluster DOOR probabilities, and we have developed estimation and inference procedures for both estimands. Although we focus on DOOR probabilities, it is straightforward to extend the proposed methodology using the influence function representation and the delta method to relative measures such as win ratio and win odds commonly used in the survival analysis literature. 
\\\\
It is worth noting that the net benefit considered in \citet{fangSampleSizeDetermination2025} is equivalent to the between-cluster DOOR probability $D_b$ introduced here. There are, however, several important differences between the work of \citet{fangSampleSizeDetermination2025} and the proposed work. First, their approach was developed specifically for parallel cluster randomized trials, whereas the proposed methods incorporate all possible design options. For example, in eye trials, some participants contribute one eye to the study, while others contribute two eyes; the two eyes of a participant (i.e., a cluster) can be randomized to the same treatment arm or different arms. Second, \citet{fangSampleSizeDetermination2025} focuses on between-cluster comparisons, whereas both within- and between-cluster comparisons are allowed in the proposed methods. As demonstrated in simulation studies, between-cluster comparisons suffer from lower powers to detect the true treatment effect in the presence of moderate or large between-cluster variability. Finally, both approaches used similar procedures to correct for small-sample biases in the standard error estimates. \citet{fangSampleSizeDetermination2025} provided simulation results for $n$ ranging from 24 to 110. In contrast, we considered more realistic situations with $n$ ranging from 10 to 20 to evaluate the performance of the small-sample correction.  
\\\\
 In the present work, we assume equal weights for each $\Phi_{ik}$, that is, for each comparison between two clusters when estimating $D_b$. Equivalently, this assigns the same weight to each individual comparison $\phi_{ik}, (i\neq k)$. To handle highly imbalanced or correlated data, \citet{larocqueTwoSampleTests2010b} introduced a function of intra-cluster dependence and cluster size, $\lambda_{ik}$, and suggested using $w_{ik} \propto \frac{1}{\bigl(1 + (m_{i1} m_{k2} - 1)\lambda_{ik}\bigr)}$ or $w_{ik}=0$ when $m_{i1}=0$ or $m_{i2}=0$. Exploring such weighting strategies within the proposed framework could further improve the estimation in the presence of substantial cluster size imbalance. 
\\\\
There are several limitations with the proposed methods. First, we focus on analyzing data from cluster-randomized trials. It would be desirable to develop methods for sample size and power calculations across all possible design options for cluster-randomized trials. Second, interim monitoring is essential in clinical trials to assess safety, efficacy, futility, and trial integrity. As a future research topic, we will extend the group-sequential designs for eye trials \citet{Diao13112025} to generic cluster-randomized trials. Specifically, we will derive the information fractions and efficacy and futility boundaries at interim looks \citet{jennison2000group}.  Third, in real applications, we assume that there are no temporal effects under the crossover design, i.e., that the outcome does not change systematically over time, independent of the treatment itself. We will develop parametric or semiparametric methods to account for temporal effects in cluster randomized crossover trials. Finally, stepped-wedge cluster randomized trials have received a lot of attention in recent years. The proposed methods in their current form are not directly applicable. An extension of the proposed methods to stepped-wedge cluster randomized trials is ongoing.

\section*{Acknowledgments}
The authors thank Dr. Madeline Rice and Laure El Ghormli for their valuable suggestions regarding the MINVI study analysis.

\newpage
\bibliographystyle{asa}
\bibliography{Reference}

\newpage
\input{main_CP_tables/Large_100_8_H1_b1.0_20260417}
\input{main_CP_tables/Small_new_10_60_CP_H1_b1.0_20260417}
\input{Suggested_Methods_v2}

\newpage
\begin{figure}[H]
    \centering 
  \includegraphics[width=0.9\textwidth]{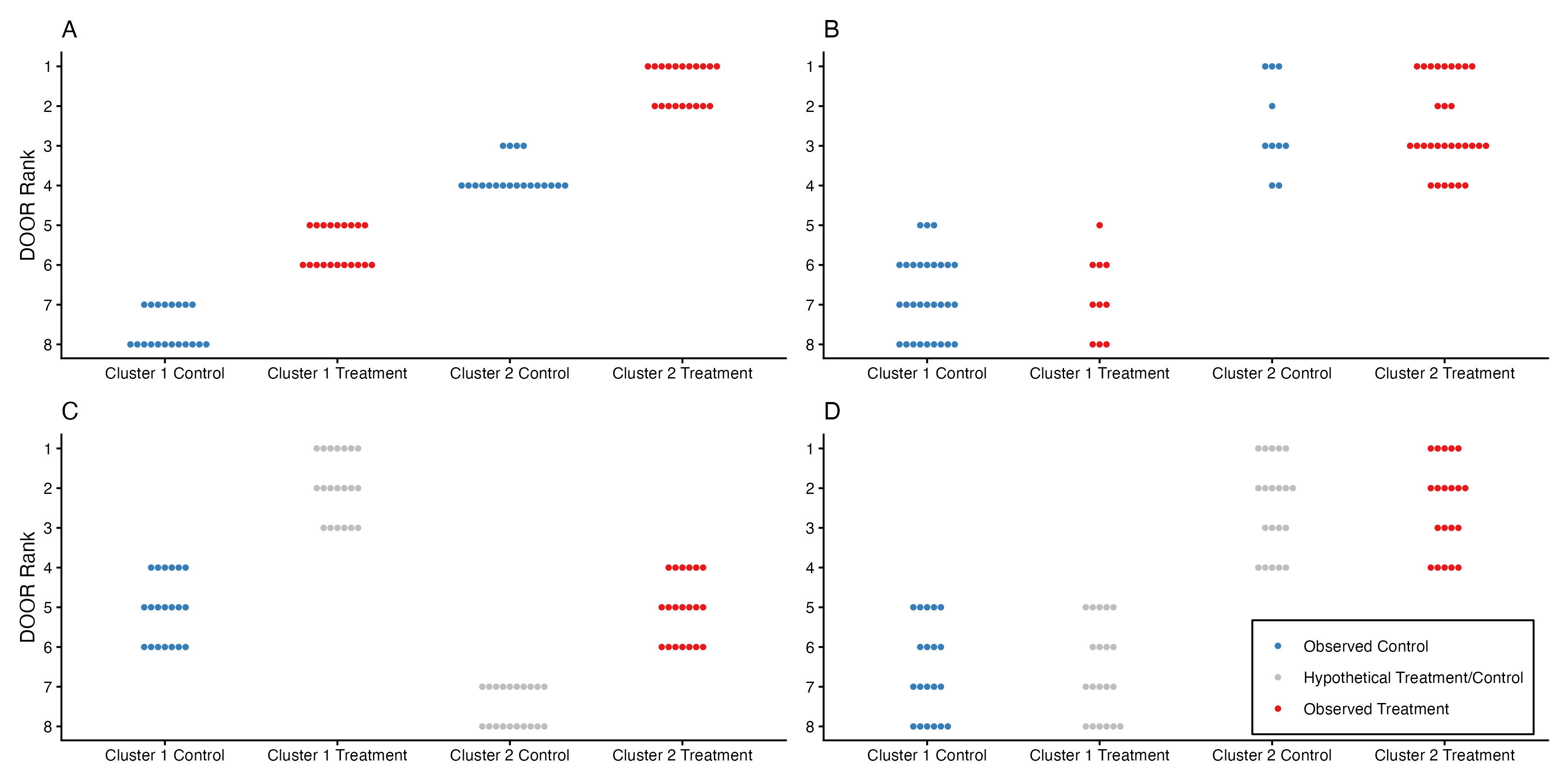}
    \caption{Examples of between-cluster variability. Each dot represents one observation. DOOR rank 1 is most desirable.} 
    \label{fig:between_clu_vari}
\end{figure}

\begin{figure}[H]
    \centering 
  \includegraphics[width=0.9\textwidth]{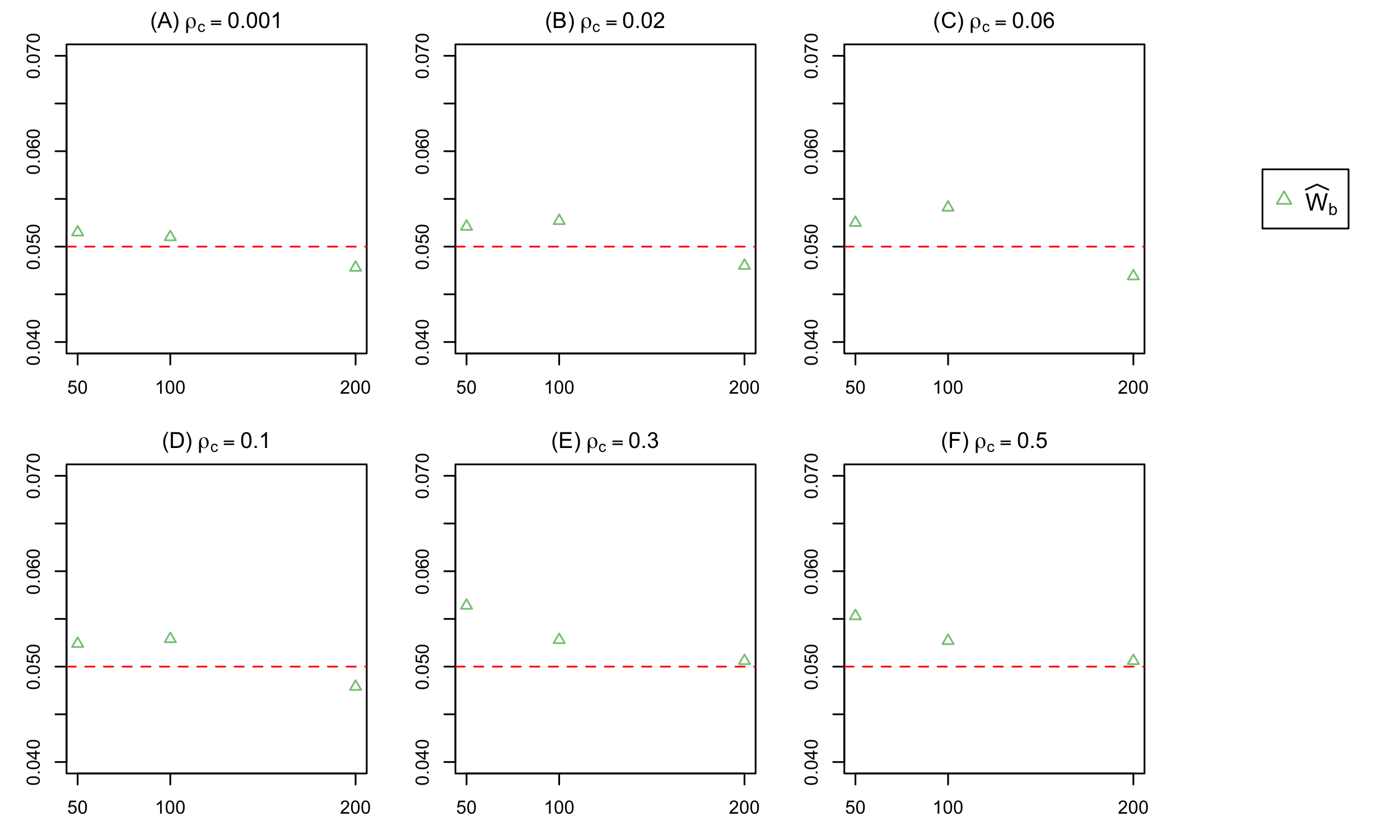}
    \caption{Type I error rates for the test based on $\widehat{W}_b$ for testing $H_0: D_b=0.5$  under the large sample size scenarios (1. $n$=50, $m$=16; 2. $n$=100, $m$=8; 3. $n$=200, $m$=4) for one-group randomization  based on  10,000 replicates. Panels (A)-(F) correspond to $\rho_c$=0.001, 0.02, 0.06, 0.1, 0.3, and 0.5, respectively. 
    } 
    \label{fig:large_1_TIE}
\end{figure}

\begin{figure}[H]
    \centering 
  \includegraphics[width=0.9\textwidth]{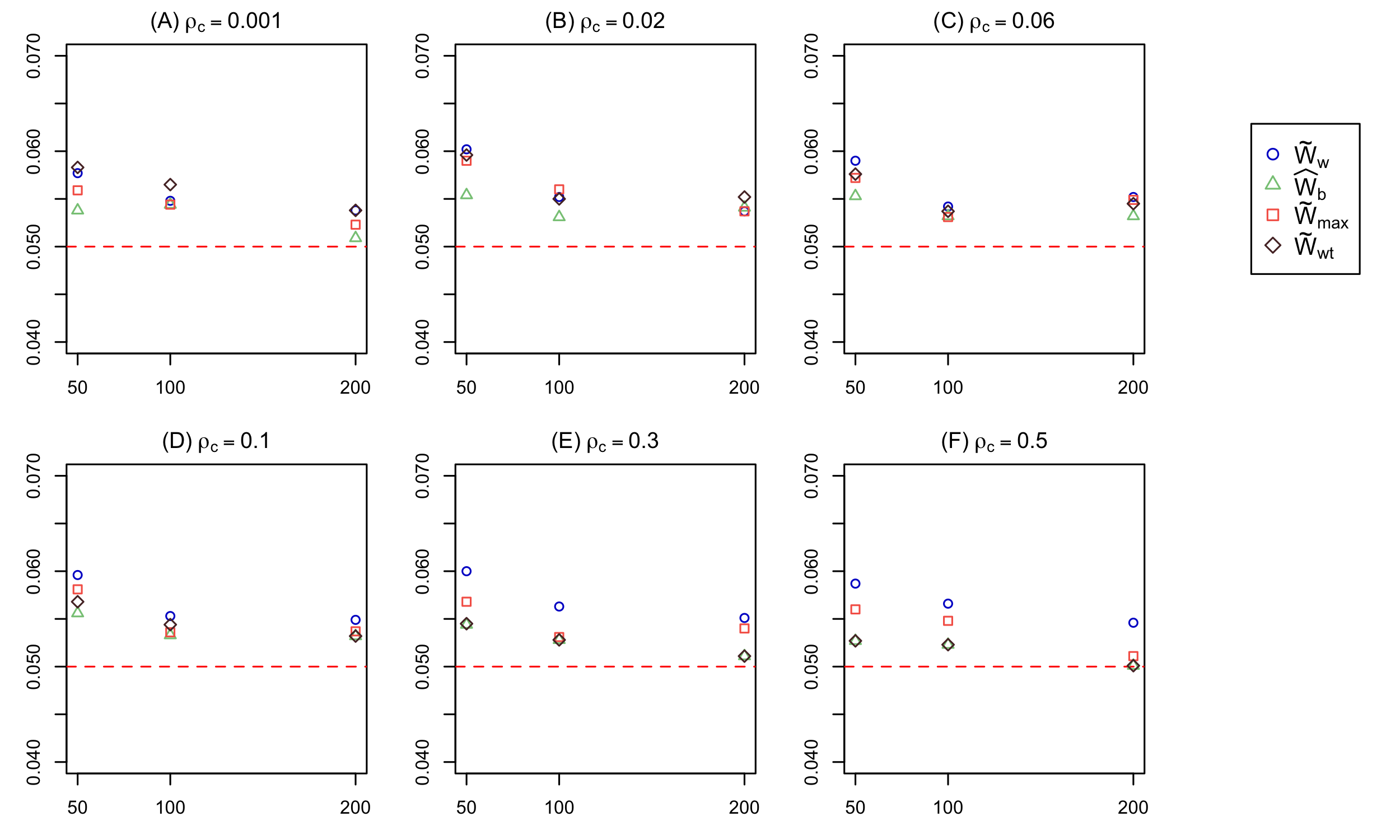}
    \caption{
    Type I error rates for the tests based on $\widetilde{W}_w$, $\widehat{W}_b$, $\widetilde{W}_{max}$ and $\widetilde{W}_{wt}$ for testing $H_0: D_w=0.5$, $H_0: D_b=0.5$, $H_0: D_b=D_w=0.5$, and $H_0: D_b=D_w=0.5$, respectively,  under the large sample size scenarios (1. $n$=50, $m$=16; 2. $n$=100, $m$=8; 3. $n$=200, $m$=4) for two-group randomization based on  10,000 replicates. Panels (A)-(F) correspond to $\rho_c$=0.001, 0.02, 0.06, 0.1, 0.3, and 0.5, respectively.
    } 
    \label{fig:large_2_TIE}
\end{figure}

\begin{figure}[H]
    \centering 
  \includegraphics[width=0.9\textwidth]{ 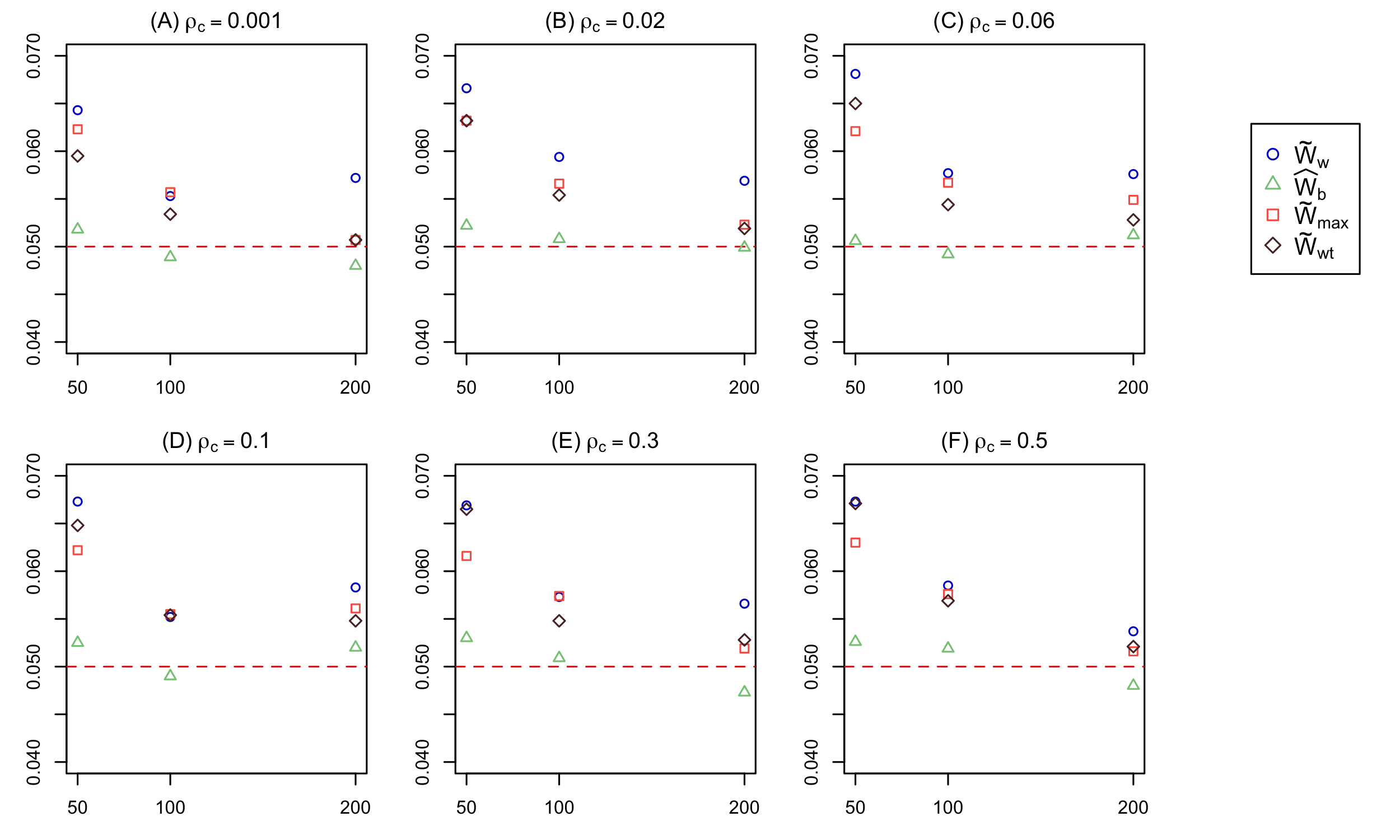}
    \caption{
    Type I error rates for the tests based on $\widetilde{W}_w$, $\widehat{W}_b$, $\widetilde{W}_{max}$ and $\widetilde{W}_{wt}$ for testing $H_0: D_w=0.5$, $H_0: D_b=0.5$, $H_0: D_b=D_w=0.5$, and $H_0: D_b=D_w=0.5$, respectively,  under the large sample size scenarios (1. $n$=50, $m$=16; 2. $n$=100, $m$=8; 3. $n$=200, $m$=4) for mixture randomization based on  10,000 replicates. Panels (A)-(F) correspond to $\rho_c$=0.001, 0.02, 0.06, 0.1, 0.3, and 0.5, respectively. 
    } 
    \label{fig:large_mix_TIE}
\end{figure}

\begin{figure}[H]
    \centering 
  \includegraphics[width=0.9\textwidth]{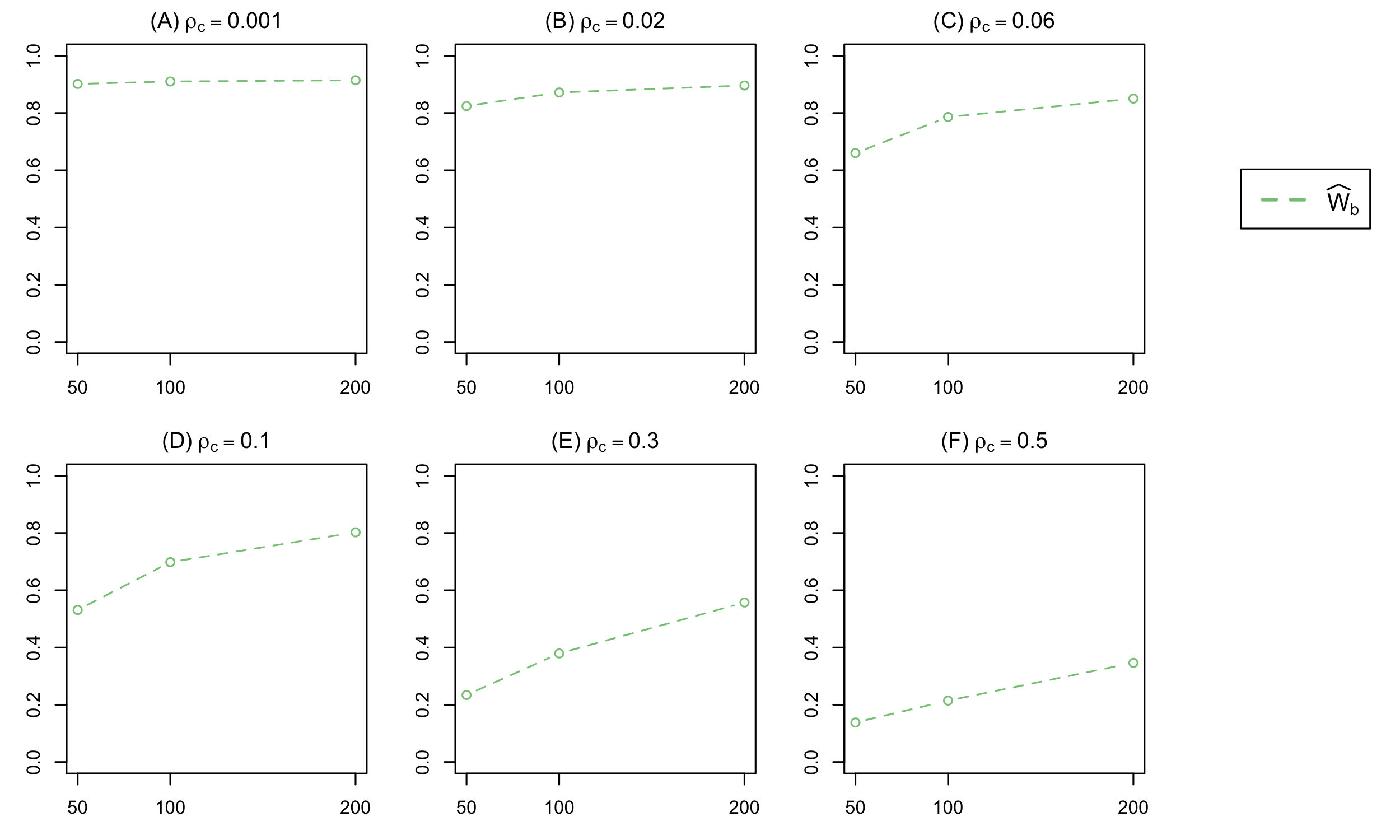}
    \caption{Powers for the test based on $\widehat{W}_b$ for testing $H_0: D_b=0.5$  under the large sample size scenarios (1. $n$=50, $m$=16; 2. $n$=100, $m$=8; 3. $n$=200, $m$=4) for one-group randomization  based on  10,000 replicates. Panels (A)-(F) correspond to $\rho_c$=0.001, 0.02, 0.06, 0.1, 0.3, and 0.5, respectively.} 
    \label{fig:large_1_power}
\end{figure}

\begin{figure}[H]
    \centering 
  \includegraphics[width=0.9\textwidth]{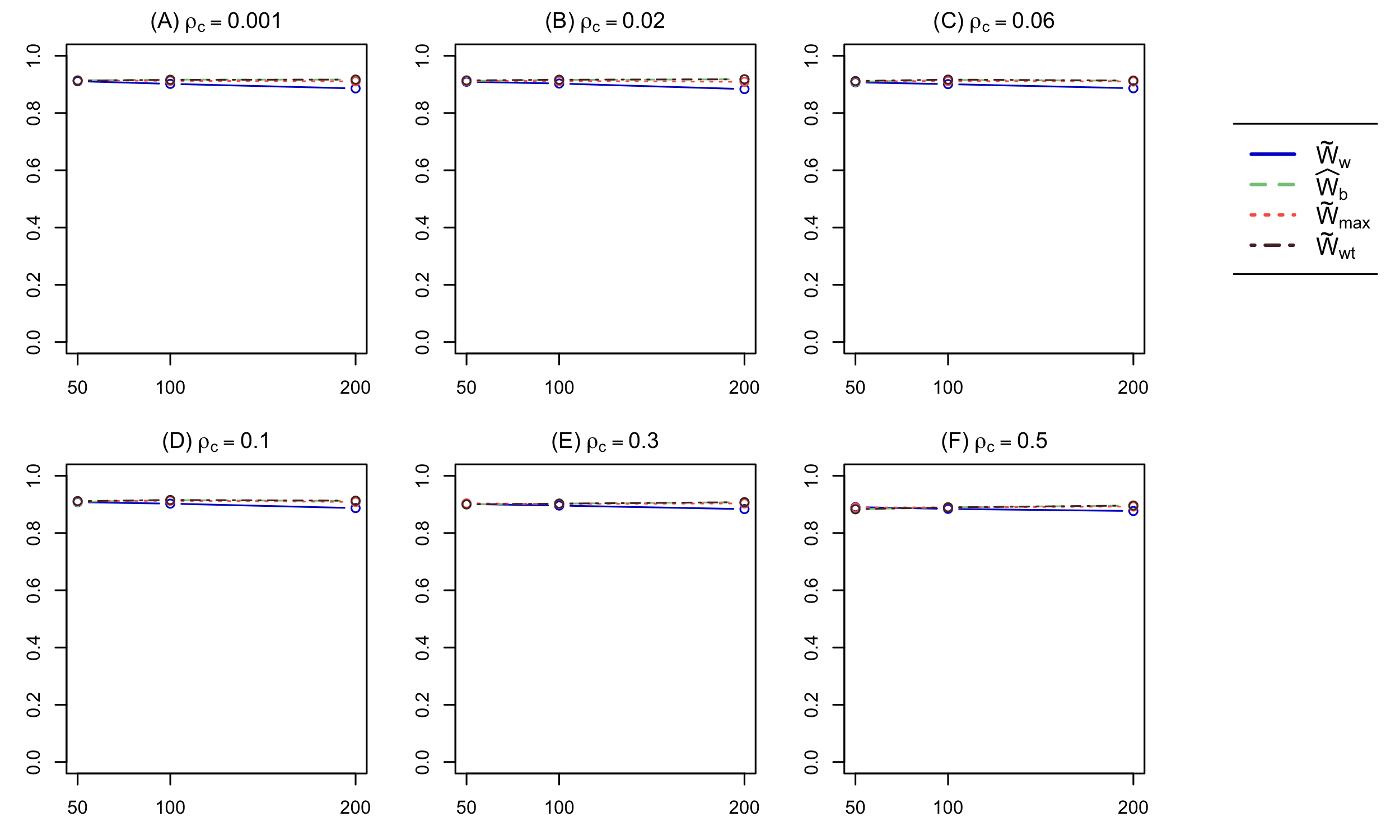}
    \caption{ Powers for the tests based on $\widetilde{W}_w$, $\widehat{W}_b$, $\widetilde{W}_{max}$ and $\widetilde{W}_{wt}$ for testing $H_0: D_w=0.5$, $H_0: D_b=0.5$, $H_0: D_b=D_w=0.5$, and $H_0: D_b=D_w=0.5$, respectively,  under the large sample size scenarios (1. $n$=50, $m$=16; 2. $n$=100, $m$=8; 3. $n$=200, $m$=4) for two-group randomization based on  10,000 replicates. Panels (A)-(F) correspond to $\rho_c$=0.001, 0.02, 0.06, 0.1, 0.3, and 0.5, respectively.} 
    \label{fig:large_2_power}
\end{figure}

\begin{figure}[H]
    \centering 
  \includegraphics[width=0.9\textwidth]{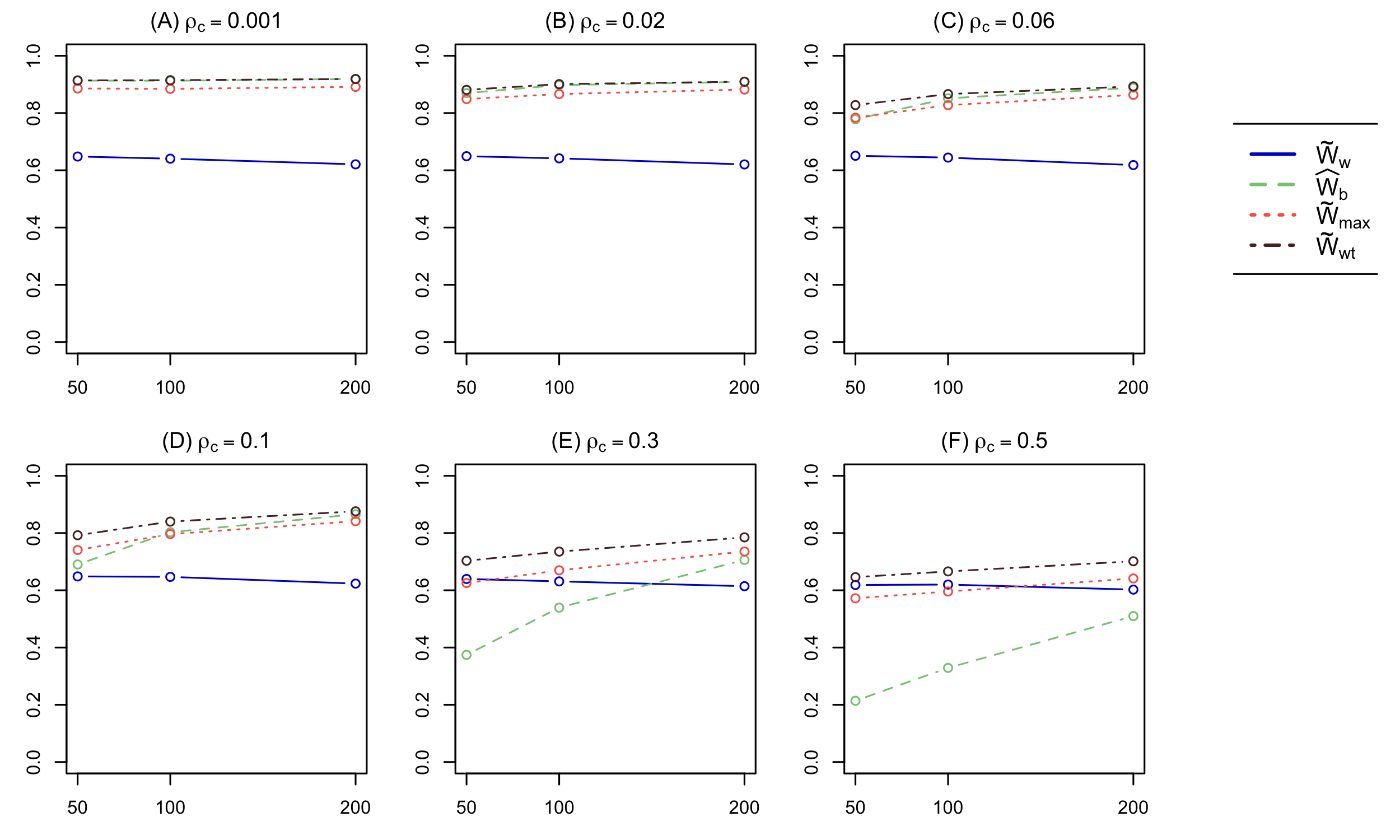}
    \caption{Powers for the tests based on $\widetilde{W}_w$, $\widehat{W}_b$, $\widetilde{W}_{max}$ and $\widetilde{W}_{wt}$ for testing $H_0: D_w=0.5$, $H_0: D_b=0.5$, $H_0: D_b=D_w=0.5$, and $H_0: D_b=D_w=0.5$, respectively,  under the large sample size scenarios (1. $n$=50, $m$=16; 2. $n$=100, $m$=8; 3. $n$=200, $m$=4) for mixture randomization based on  10,000 replicates. Panels (A)-(F) correspond to $\rho_c$=0.001, 0.02, 0.06, 0.1, 0.3, and 0.5, respectively.} 
    \label{fig:large_mix_power}
\end{figure}

\begin{figure}[H]
    \centering 
    \includegraphics[width=0.9\textwidth]{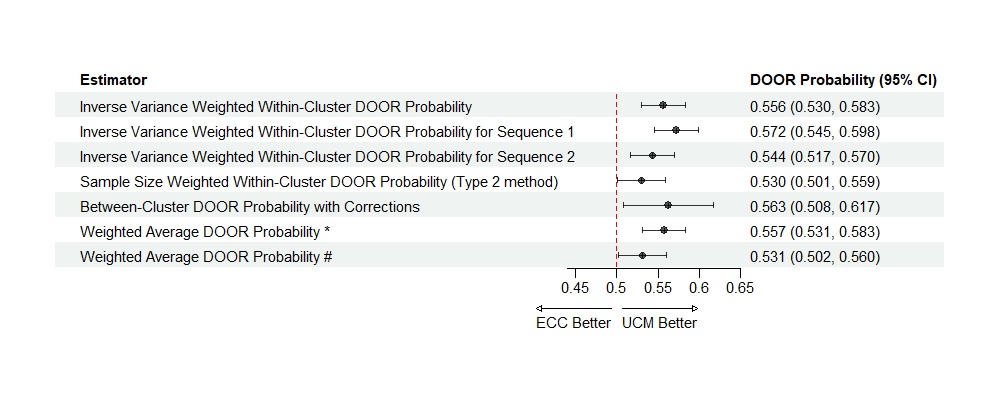}
    \caption{Results of the MINVI study DOOR analysis. $^*$ Weighted average of $\widehat{D}_w$ and $\widehat{D}_b$ with small-sample corrections. $^{\#}$  Weighted average of $\widetilde{D}_w$ (Type 2 method) and $\widehat{D}_b$ with corrections. } 
    \label{fig:MINVI_study}
\end{figure}

\newpage
\setcounter{page}{1}
\begin{center}
    {\Large\bf Supplement to ``On Cluster Randomized Trials with the Desirability of Outcome Ranking (DOOR) Endpoints"}
\end{center}
\section*{Supplement A}
Here we provide expressions of $\sigma_{10i}^2$ and $\sigma_{01i}^2$ described in section 2.1.
\[
    \sigma^2_{10i}
= \sum_{j_k=1}^{K-1} p_{1i,j_k}
\left(\sum_{l_k=j_k}^{K} p_{2i,l_k}\right)
\left(\sum_{l_k=j_k+1}^{K} p_{2i,l_k}\right)
+ \frac{1}{4}\sum_{j_k=1}^{K} p_{1i,j_k} \, p_{2i,j_k}^{2}
- D_{wi}^2
\]
where 
$p_{1i,j_k} = P(Y_{ij}=j_k \mid A_{ij}=1),
p_{2i,l_k} = P(Y_{il}=l_k \mid A_{il}=0),$
for \(j_k, l_k = 1,2,\ldots,K\). Similarly,
\[
\sigma_{01i}^2
=
\sum_{l_k=1}^{K-1} p_{2i,l_k}
\left(\sum_{j_k=l_k}^{K} p_{1i,j_k}\right)
\left(\sum_{j_k=l_k+1}^{K} p_{1i,j_k}\right)
+ \frac{1}{4}\sum_{l_k=1}^{K} p_{2i,l_k}\, p_{1i,l_k}^{2}
-\left(1 - D_{wi}\right)^2.
\]
The probability of a random selected patient $j$ in cluster $i$ from treatment group stays at DOOR level $j_k$ can be estimated by the observed frequencies $\hat{p}_{1i,j_k}=\sum_{j=1}^{m_{i1}} I(Y_{ij}=j_k)/m_{i1}$.
\citet{shuLongitudinalBenefitRisk2025} developed a framework for longitudinal DOOR, and they have proved the asymptotic variance of the estimator of the DOOR probability at a time point in the Appendix, which can be viewed as the asymptotic variance of the DOOR probability for a cluster here.

\renewcommand\thefigure{S\arabic{figure}}\setcounter{figure}{0}
\renewcommand\thetable{S\arabic{table}}\setcounter{table}{0}

\section*{Supplement B}
\begin{proof}
Since $\widehat D_b$ is a U-statistic with kernel $h(\boldsymbol{Y}_i,\boldsymbol{Y}_k)$, we can use the
projection method for U-statistics. The projection of $\widehat D_b-D_b$ onto the
space of statistics of the form $\sum_{i=1}^n g(\boldsymbol{Y}_i)$ is given by the Hájek
projection
$\widehat U
=
\sum_{i=1}^n E[\widehat D_b - D_b \mid \boldsymbol{Y}_i]$.
\\
Because $\widehat D_b = \binom{n}{2}^{-1}\sum_{i<k} h(\boldsymbol{Y}_i,\boldsymbol{Y}_k)$, conditioning on
$\boldsymbol{Y}_i$ retains only the terms involving $\boldsymbol{Y}_i$. Hence
\[
E[\widehat D_b-D_b \mid \boldsymbol{Y}_i]
=
\frac{2}{n}\,
E[h(\tilde{\boldsymbol{Y}}_i,\boldsymbol{Y}_i)-D_b \mid \boldsymbol{Y}_i].
\] Define
$\psi_{bi}=2\,E[h(\tilde{\boldsymbol{Y}}_i,\boldsymbol{Y}_i)-D_b \mid \boldsymbol{Y}_i].$
Then $\psi_{bi}$ is the influence function of $D_b$, and the projection becomes
$\widehat U=\frac{1}{n}\sum_{i=1}^n\psi_{bi}.$
Since $\psi_{bi}$ are i.i.d. with mean zero, the central limit theorem yields
\[
\sqrt{n}\widehat U
=
\frac{1}{\sqrt n}\sum_{i=1}^n \psi_{bi}
\xrightarrow{d}
N(0,\text{Var}(\psi_{bi})).
\]
The difference between $\widehat D_b-D_b$ and its projection $\widehat{U}=\frac{1}{n}\sum_{i=1}^n\psi_{bi}$ is asymptotically negligible,
\[
\sqrt{n}(\widehat D_b-D_b-\widehat U) \xrightarrow{p} 0.
\]
Therefore,
\[
\sqrt n(\widehat D_b-D_b)
=
\frac{1}{\sqrt n}\sum_{i=1}^n \psi_{bi} + o_p(1).
\]
By Slutsky's theorem,
\[
\sqrt n(\widehat D_b-D_b)
\xrightarrow{d}
N(0,\text{Var}(\psi_{bi})).
\]
\end{proof}

\section*{Supplement C.1}
\textbf{Large sample size scenarios.}
\input{supp_CP_tables/Large_100_8_H1_b1.0_corr_20260417}
\input{supp_CP_tables/Large_50_16_H1_b1.0_20260417}
\input{supp_CP_tables/Large_200_4_H1_b1.0_20260417}

\newpage
\textbf{Small-sample size scenarios without corrections.}
\input{supp_CP_tables/Small_old_10_60_CP_H1_b1.0_20260417}
\input{supp_CP_tables/Small_old_15_40_CP_H1_b1.0_20260417}
\input{supp_CP_tables/Small_old_20_30_CP_H1_b1.0_20260417}

\newpage
\textbf{Small-sample size scenarios with corrections or adjustments.}
\input{supp_CP_tables/Small_new_15_40_CP_H1_b1.0_20260417}
\input{supp_CP_tables/Small_new_20_30_CP_H1_b1.0_20260417}

\newpage
\section*{Supplement C.2}
\textbf{Small-sample size scenarios without corrections } \\
\textbf{Type I error rates.}
\begin{figure}[H]
    \centering 
  \includegraphics[width=0.9\textwidth]{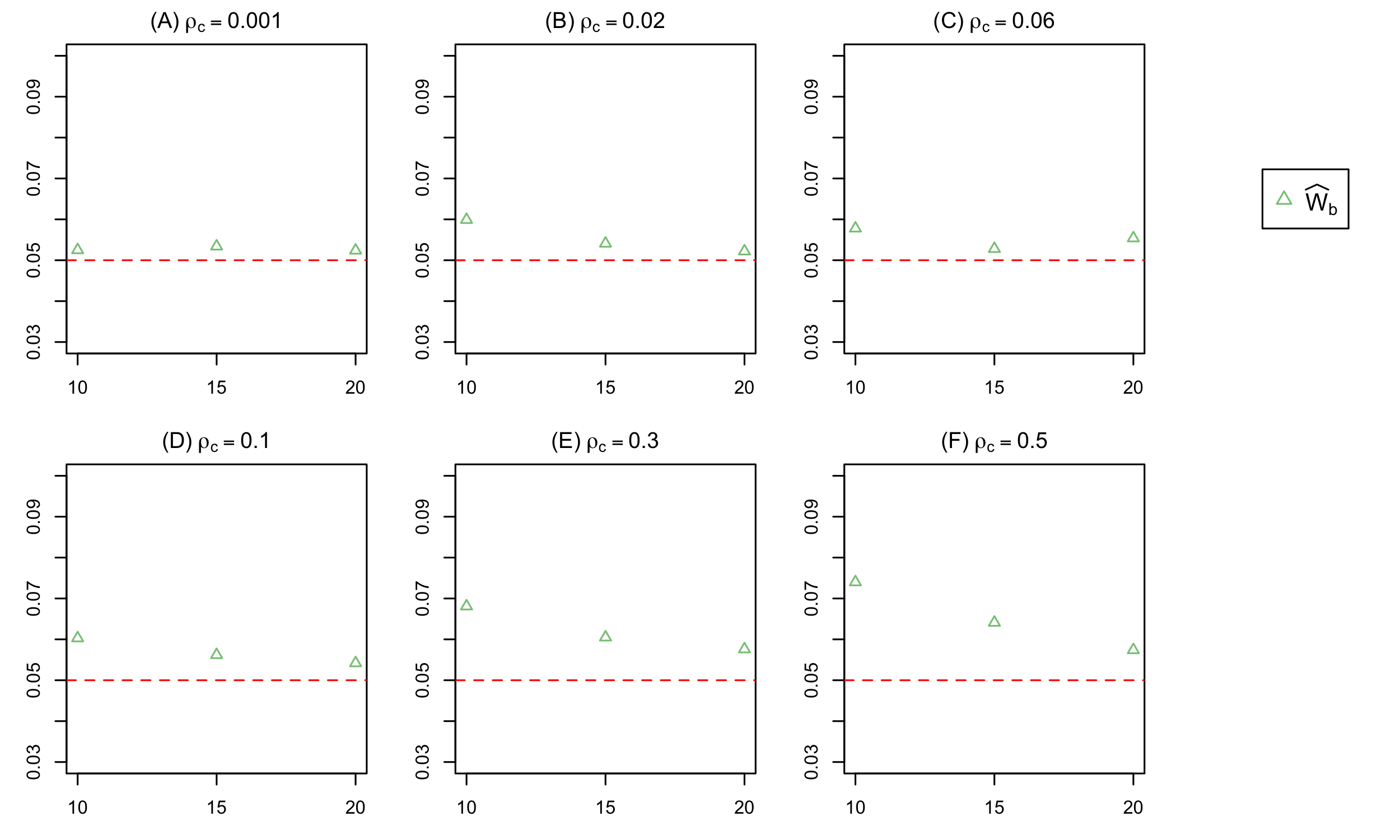}
    \caption{Type I error rates for the test based on $\widehat{W}_b$ without small-sample corrections for testing $H_0: D_b=0.5$  under the small-sample size scenarios (1. $n$=10, $m$=60; 2. $n$=15, $m$=40; 3. $n$=20, $m$=30) for one-group randomization  based on  10,000 replicates. Panels (A)-(F) correspond to $\rho_c$=0.001, 0.02, 0.06, 0.1, 0.3, and 0.5, respectively.} 
    \label{fig:centered}
\end{figure}

\begin{figure}[H]
    \centering 
  \includegraphics[width=0.9\textwidth]{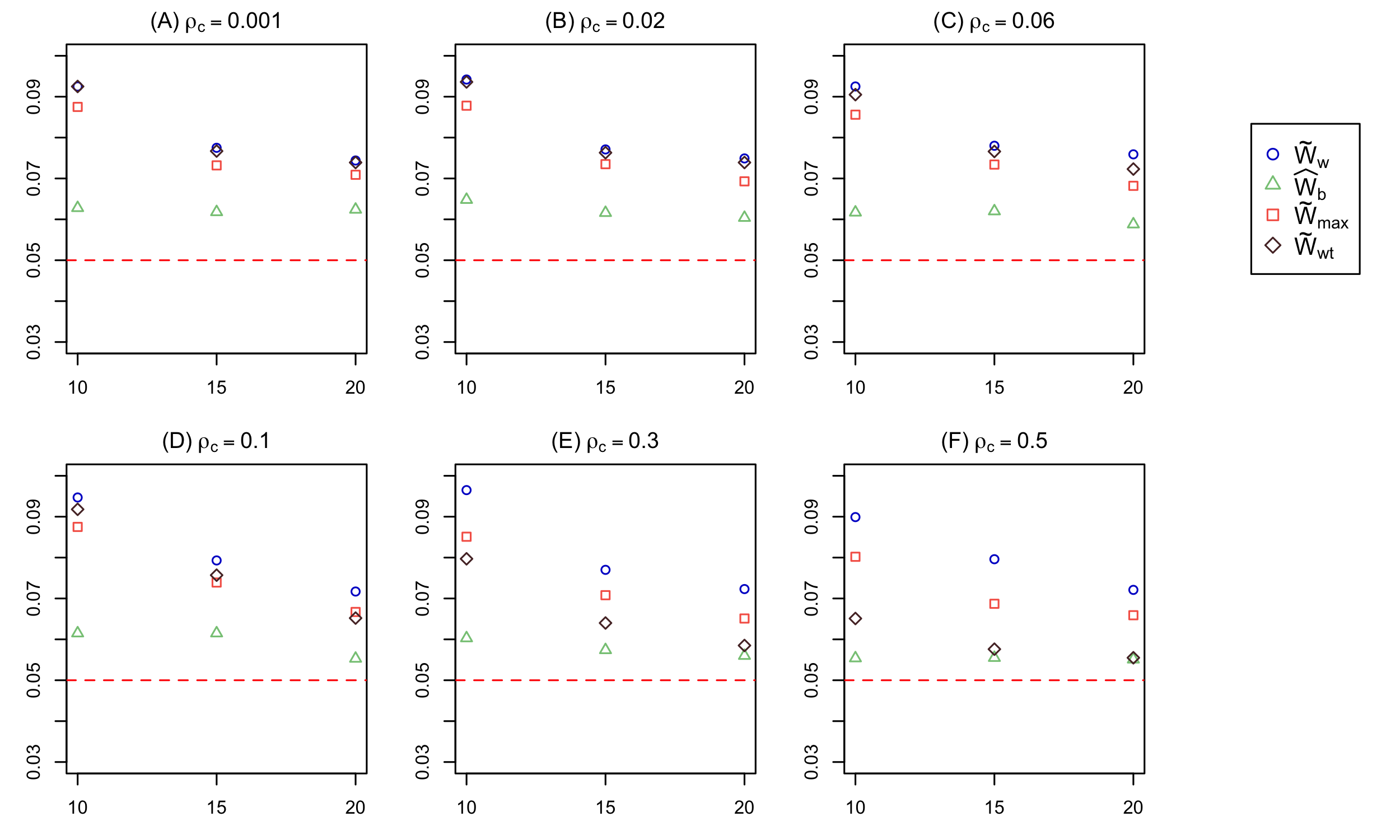}
    \caption{Type I error rates for the tests based on $\widetilde{W}_w$, $\widehat{W}_b$, $\widetilde{W}_{max}$ and $\widetilde{W}_{wt}$ without small-sample corrections for testing $H_0: D_w=0.5$, $H_0: D_b=0.5$, $H_0: D_b=D_w=0.5$, and $H_0: D_b=D_w=0.5$, respectively,  under the small-sample size scenarios (1. $n$=10, $m$=60; 2. $n$=15, $m$=40; 3. $n$=20, $m$=30) for two-group randomization based on 10,000 replicates. Panels (A)-(F) correspond to $\rho_c$=0.001, 0.02, 0.06, 0.1, 0.3, and 0.5, respectively.} 
    \label{fig:centered}
\end{figure}

\begin{figure}[H]
    \centering 
  \includegraphics[width=0.9\textwidth]{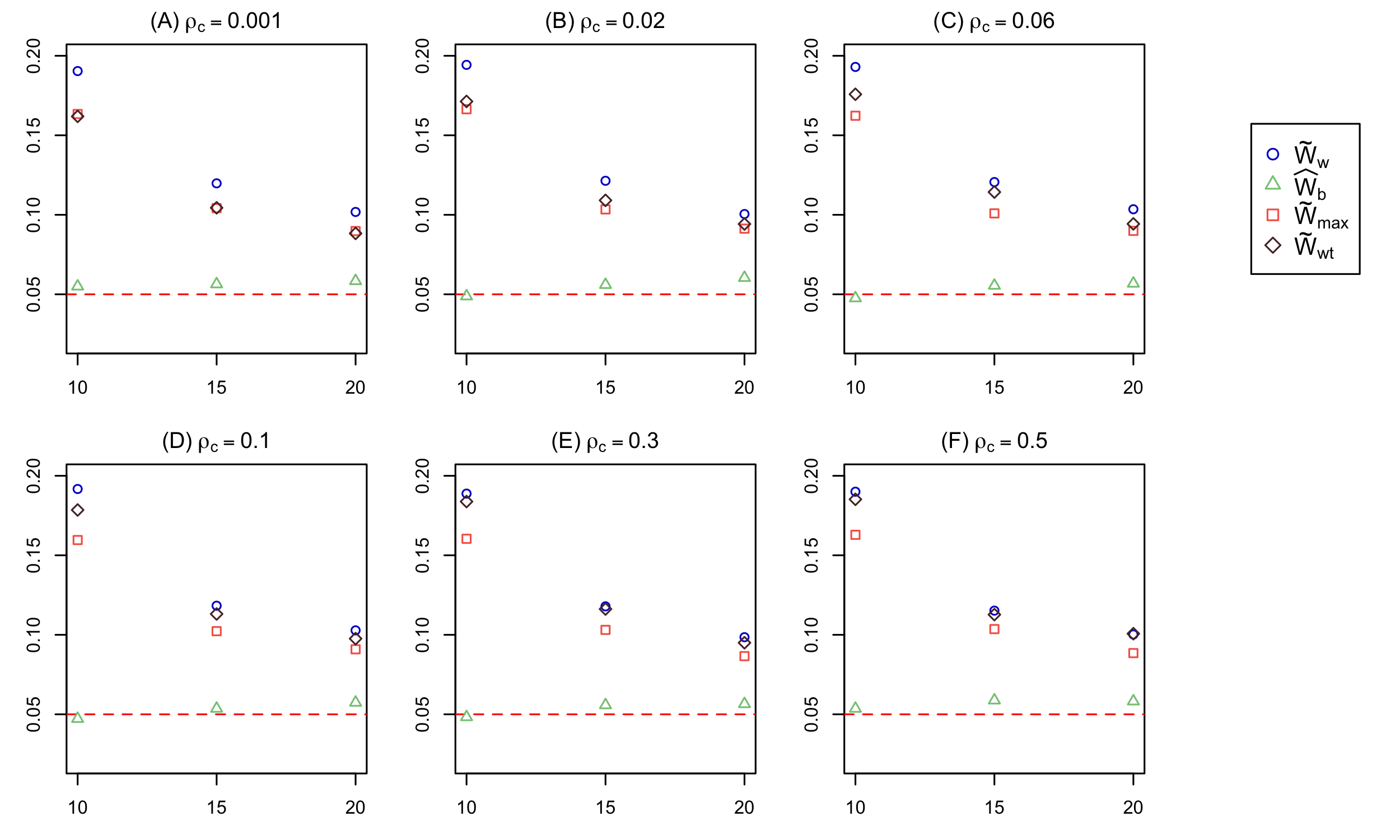}
    \caption{Type I error rates for the tests based on $\widetilde{W}_w$, $\widehat{W}_b$, $\widetilde{W}_{max}$ and $\widetilde{W}_{wt}$ without small-sample corrections for testing $H_0: D_w=0.5$, $H_0: D_b=0.5$, $H_0: D_b=D_w=0.5$, and $H_0: D_b=D_w=0.5$, respectively,  under the small-sample size scenarios (1. $n$=10, $m$=60; 2. $n$=15, $m$=40; 3. $n$=20, $m$=30) for mixture randomization based on 10,000 replicates. Panels (A)-(F) correspond to $\rho_c$=0.001, 0.02, 0.06, 0.1, 0.3, and 0.5, respectively.} 
    \label{fig:centered}
\end{figure}

\newpage
\noindent  \textbf{Small-sample size scenarios with corrections or adjustments.} \\
\textbf{Type I error rates.}
\begin{figure}[H]
    \centering 
  \includegraphics[width=0.9\textwidth]{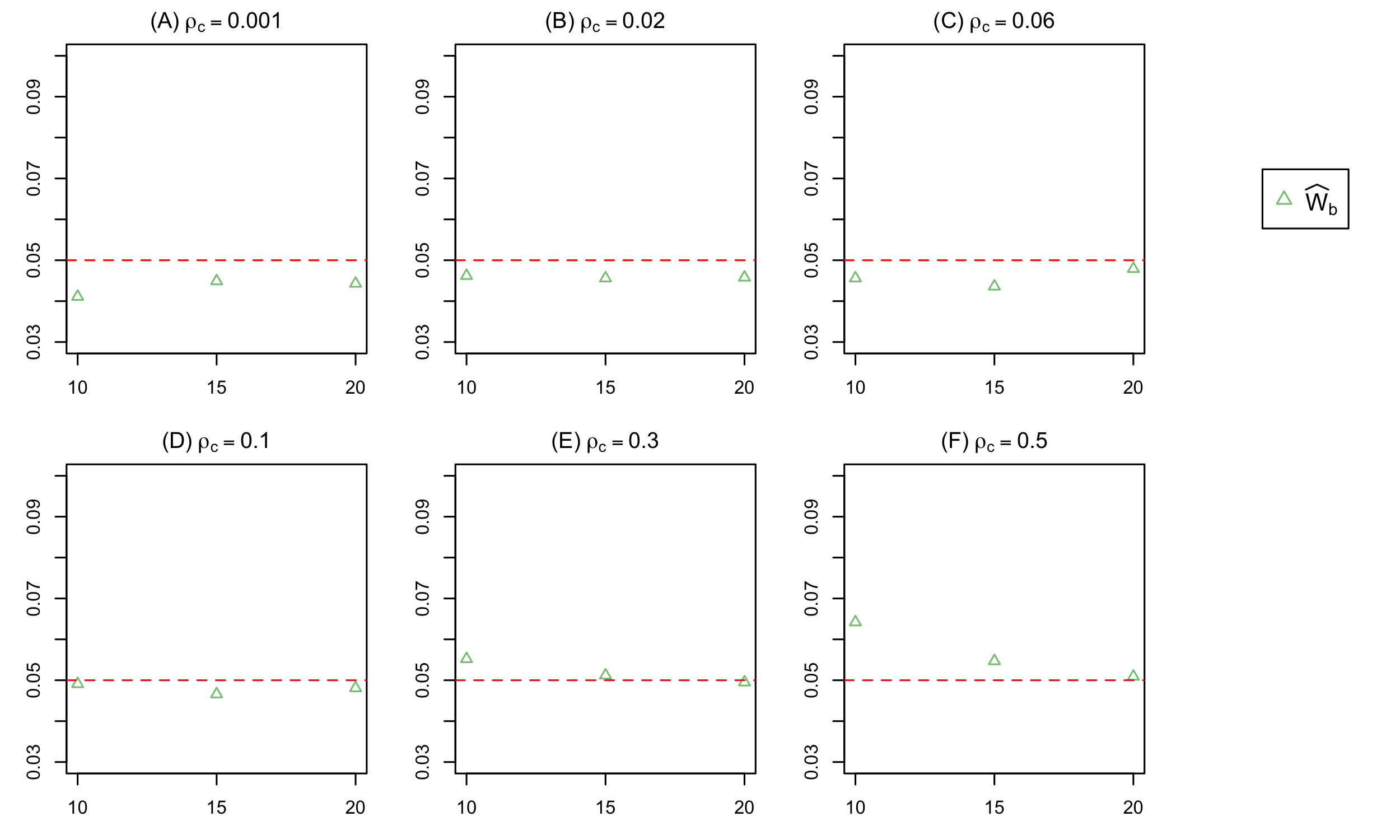}
    \caption{Type I error rates for the test based on $\widehat{W}_b$ with small-sample corrections for testing $H_0: D_b=0.5$  under the small-sample size scenarios (1. $n$=10, $m$=60; 2. $n$=15, $m$=40; 3. $n$=20, $m$=30) for one-group randomization  based on  10,000 replicates. Panels (A)-(F) correspond to $\rho_c$=0.001, 0.02, 0.06, 0.1, 0.3, and 0.5, respectively.
} 
    \label{fig:centered}
\end{figure}

\begin{figure}[H]
    \centering 
  \includegraphics[width=0.9\textwidth]{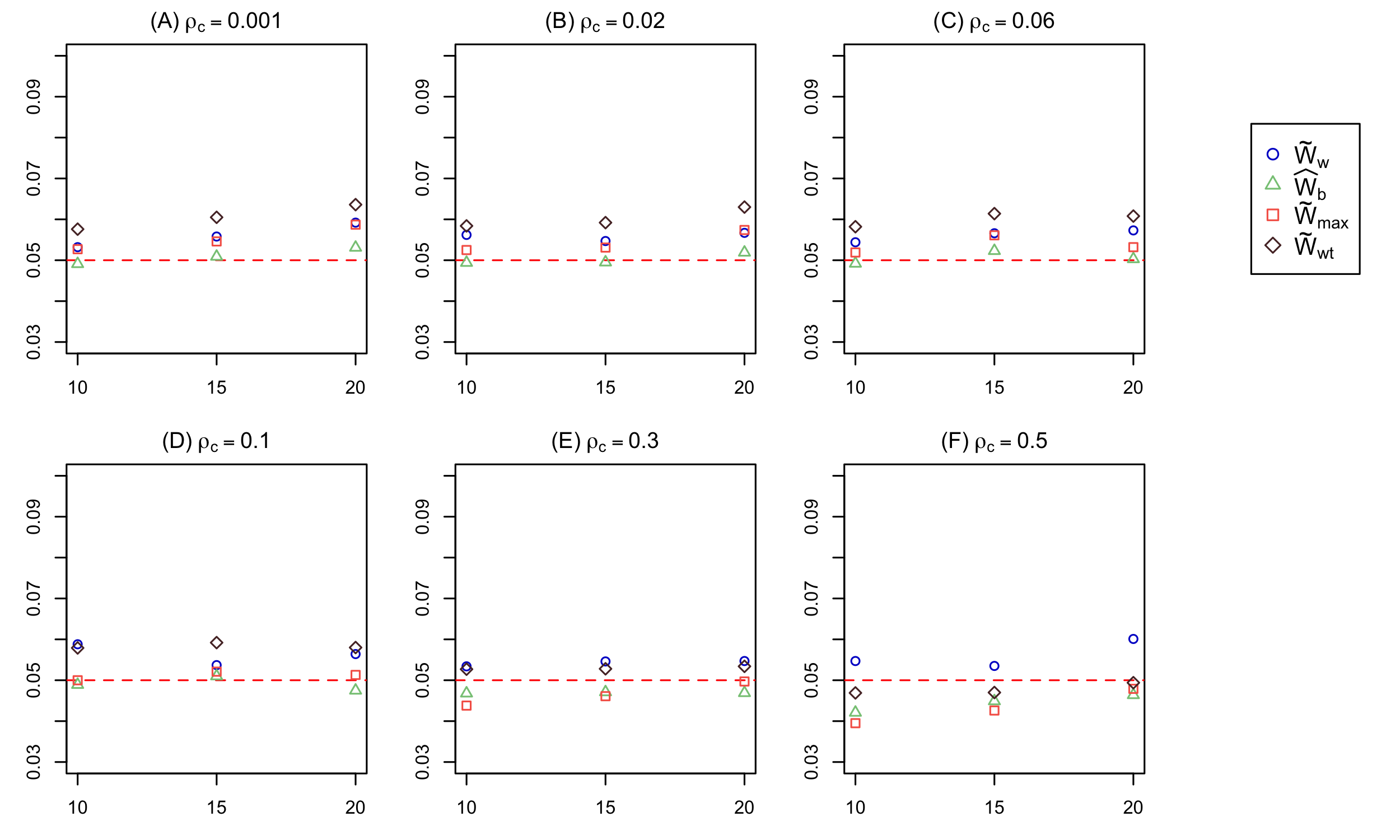}
    \caption{Type I error rates for the tests based on $\widetilde{W}_w$ with Type 2 method, $\widehat{W}_b$, $\widetilde{W}_{max}$ and $\widetilde{W}_{wt}$ with small-sample corrections or adjustments for testing $H_0: D_w=0.5$, $H_0: D_b=0.5$, $H_0: D_b=D_w=0.5$, and $H_0: D_b=D_w=0.5$, respectively,  under the small-sample size scenarios (1. $n$=10, $m$=60; 2. $n$=15, $m$=40; 3. $n$=20, $m$=30) for two-group randomization based on 10,000 replicates. Panels (A)-(F) correspond to $\rho_c$=0.001, 0.02, 0.06, 0.1, 0.3, and 0.5, respectively.} 
    \label{fig:centered}
\end{figure}

\begin{figure}[H]
    \centering 
  \includegraphics[width=0.9\textwidth]{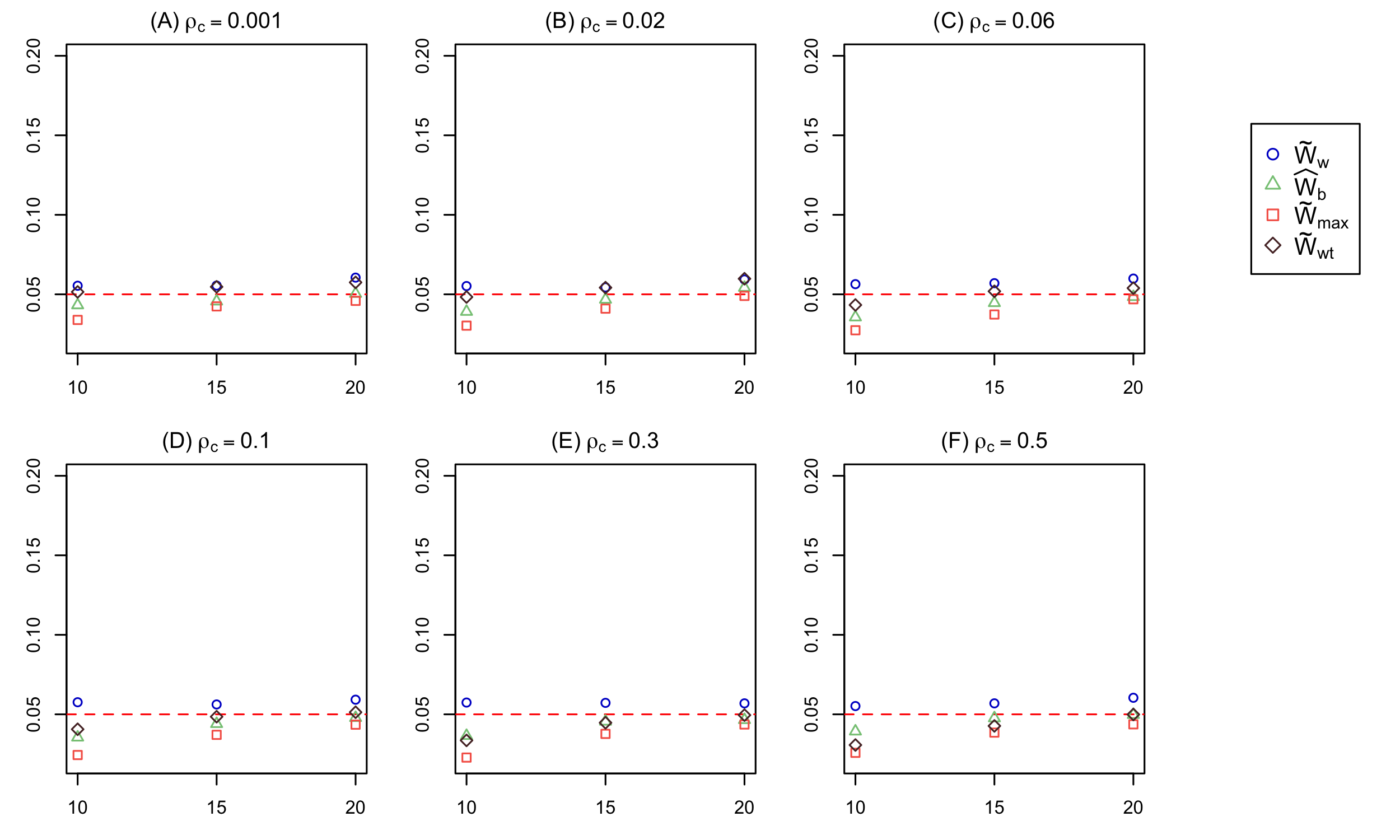}
    \caption{Type I error rates for the tests based on $\widetilde{W}_w$ with Type 2 method, $\widehat{W}_b$, $\widetilde{W}_{max}$ and $\widetilde{W}_{wt}$ with small-sample corrections or adjustments for testing $H_0: D_w=0.5$, $H_0: D_b=0.5$, $H_0: D_b=D_w=0.5$, and $H_0: D_b=D_w=0.5$, respectively,  under the small-sample size scenarios (1. $n$=10, $m$=60; 2. $n$=15, $m$=40; 3. $n$=20, $m$=30) for mixture randomization based on 10,000 replicates. Panels (A)-(F) correspond to $\rho_c$=0.001, 0.02, 0.06, 0.1, 0.3, and 0.5, respectively.} 
    \label{fig:centered}
\end{figure}

\newpage
\noindent \textbf{Small-sample size scenarios without corrections.} \\
\textbf{Powers.}
\begin{figure}[H]
    \centering 
  \includegraphics[width=0.9\textwidth]{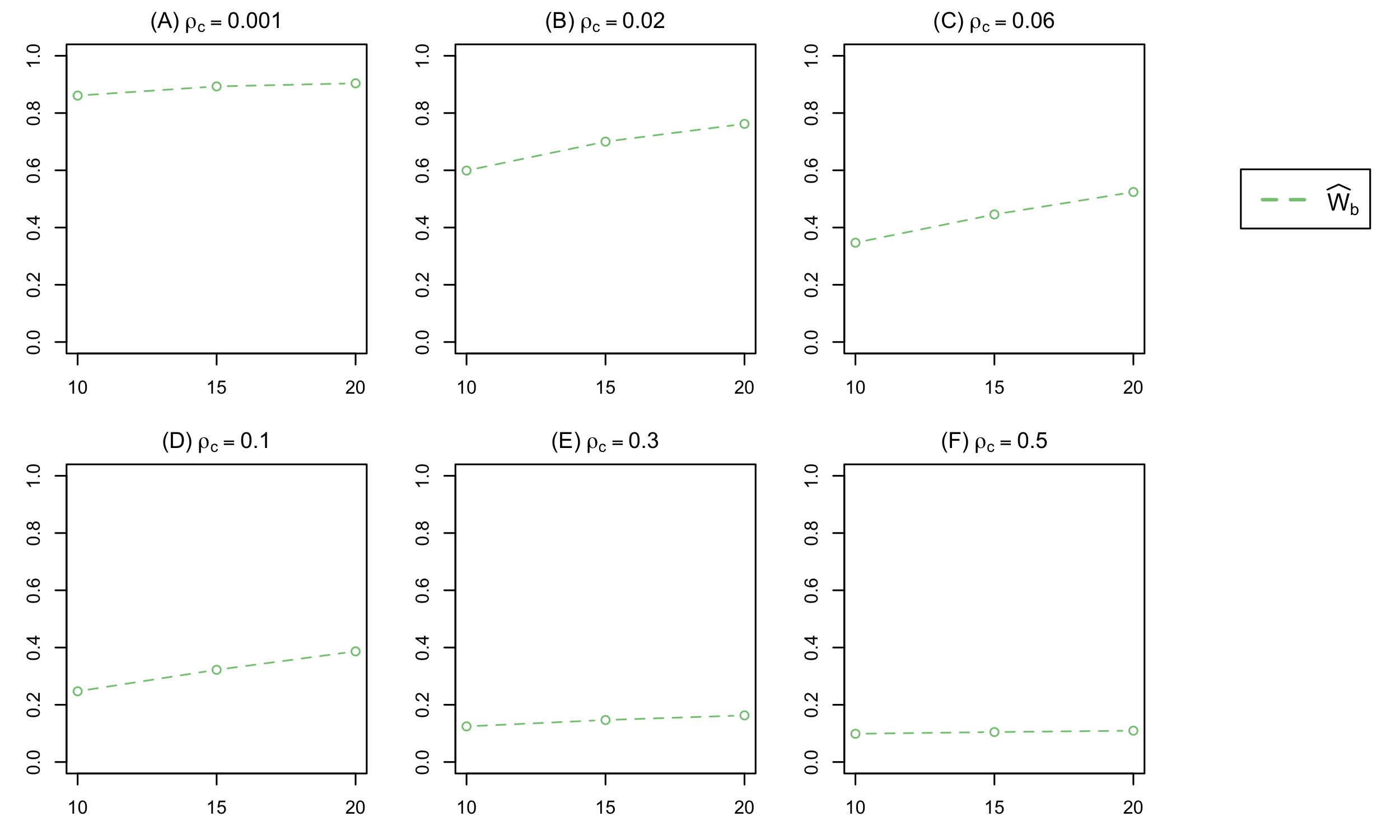}
    \caption{Powers for the test based on $\widehat{W}_b$ without small-sample corrections for testing $H_0: D_b=0.5$  under the small-sample size scenarios (1. $n$=10, $m$=60; 2. $n$=15, $m$=40; 3. $n$=20, $m$=30) for one-group randomization  based on  10,000 replicates. Panels (A)-(F) correspond to $\rho_c$=0.001, 0.02, 0.06, 0.1, 0.3, and 0.5, respectively.} 
    \label{fig:centered}
\end{figure}

\begin{figure}[H]
    \centering 
  \includegraphics[width=0.9\textwidth]{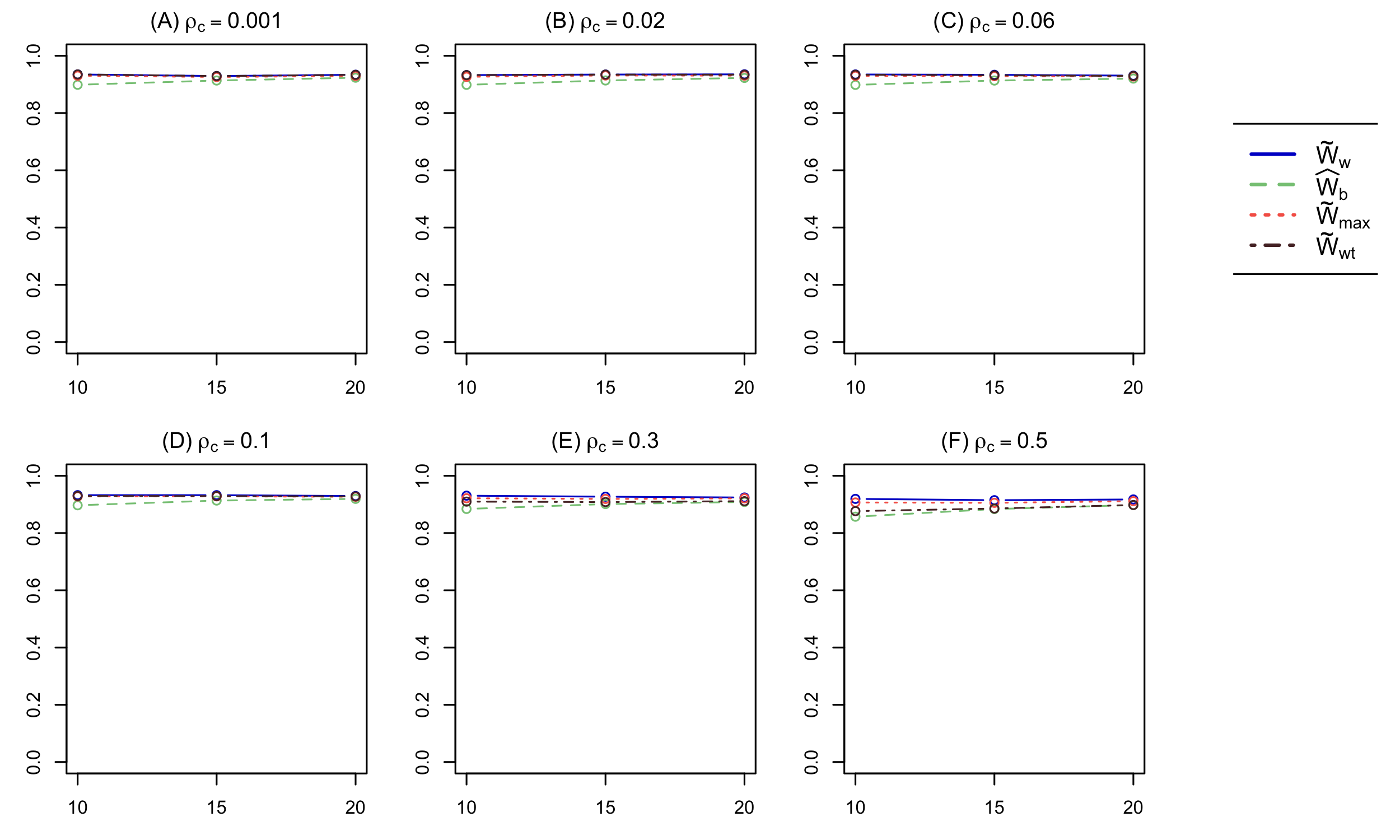}
    \caption{Powers for the tests based on $\widetilde{W}_w$, $\widehat{W}_b$, $\widetilde{W}_{max}$ and $\widetilde{W}_{wt}$ without small-sample corrections for testing $H_0: D_w=0.5$, $H_0: D_b=0.5$, $H_0: D_b=D_w=0.5$, and $H_0: D_b=D_w=0.5$, respectively,  under the small-sample size scenarios (1. $n$=10, $m$=60; 2. $n$=15, $m$=40; 3. $n$=20, $m$=30) for two-group randomization based on 10,000 replicates. Panels (A)-(F) correspond to $\rho_c$=0.001, 0.02, 0.06, 0.1, 0.3, and 0.5, respectively.} 
    \label{fig:centered}
\end{figure}

\begin{figure}[H]
    \centering 
  \includegraphics[width=0.9\textwidth]{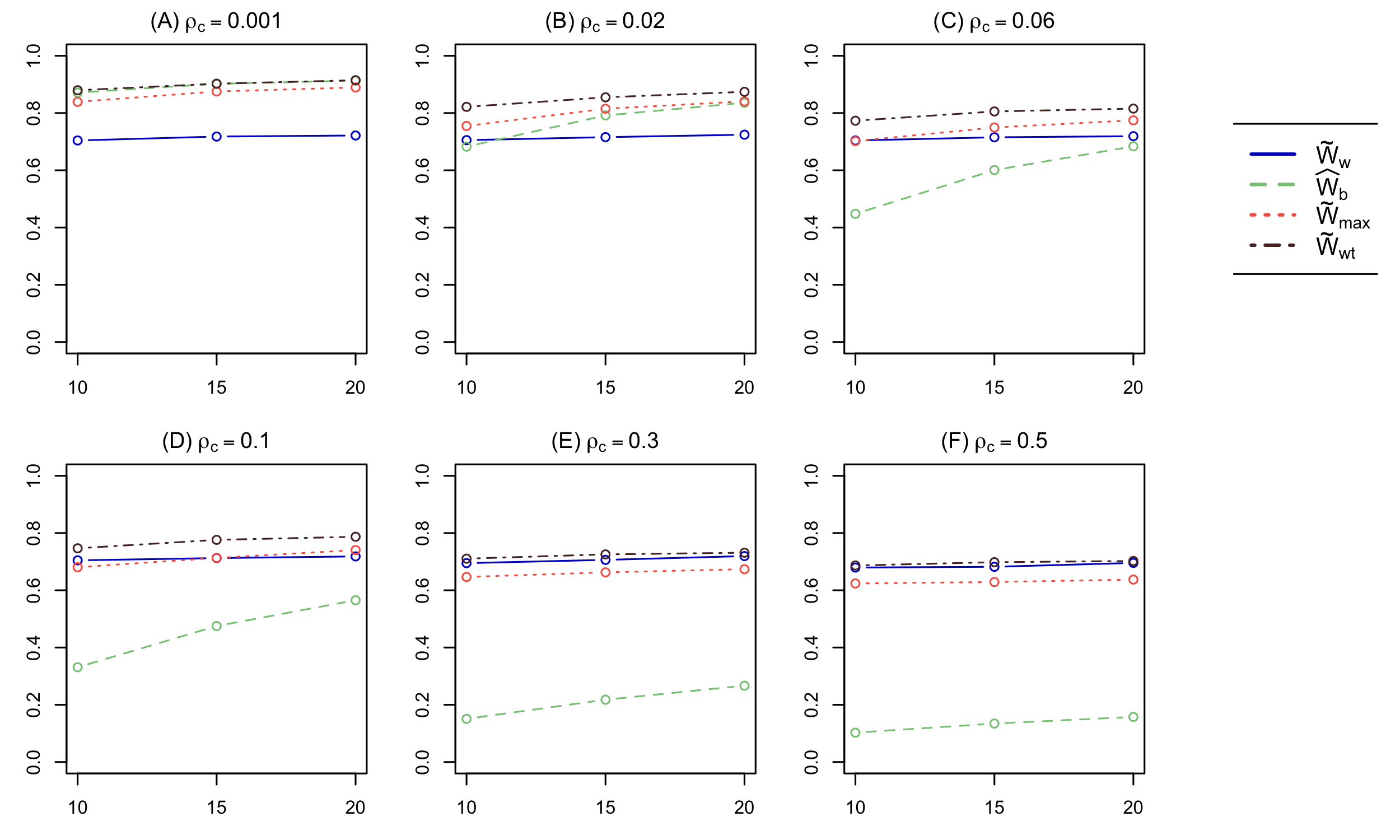}
    \caption{Powers for the tests based on $\widetilde{W}_w$, $\widehat{W}_b$, $\widetilde{W}_{max}$ and $\widetilde{W}_{wt}$ without small-sample corrections for testing $H_0: D_w=0.5$, $H_0: D_b=0.5$, $H_0: D_b=D_w=0.5$, and $H_0: D_b=D_w=0.5$, respectively,  under the small-sample size scenarios (1. $n$=10, $m$=60; 2. $n$=15, $m$=40; 3. $n$=20, $m$=30) for mixture randomization based on 10,000 replicates. Panels (A)-(F) correspond to $\rho_c$=0.001, 0.02, 0.06, 0.1, 0.3, and 0.5, respectively.} 
    \label{fig:centered}
\end{figure}

\newpage
\noindent \textbf{Small-sample size scenarios with corrections or adjustments.} \\
\textbf{Powers.}
\begin{figure}[H]
    \centering 
  \includegraphics[width=0.9\textwidth]{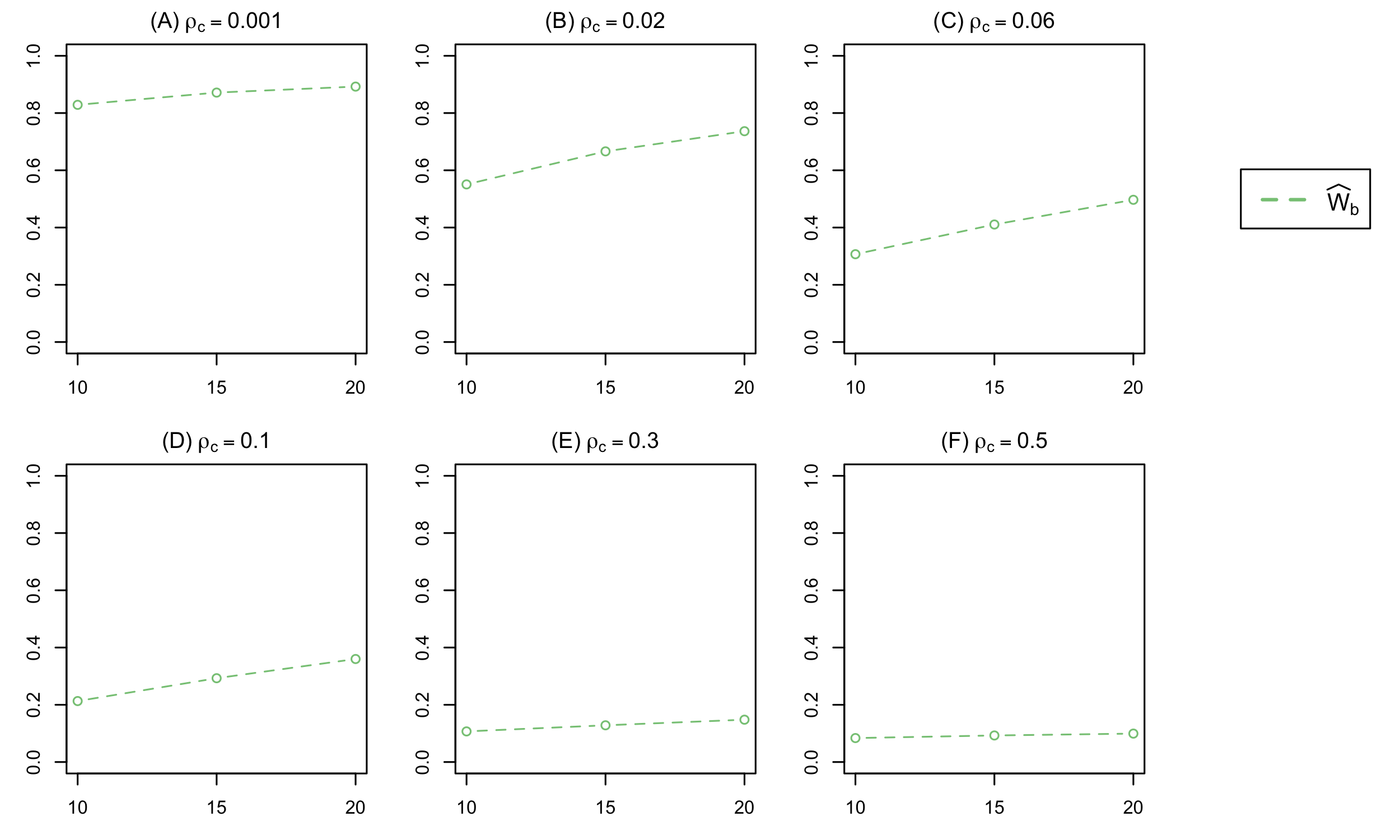}
    \caption{Powers for the test based on $\widehat{W}_b$ with small-sample corrections for testing $H_0: D_b=0.5$  under the small-sample size scenarios (1. $n$=10, $m$=60; 2. $n$=15, $m$=40; 3. $n$=20, $m$=30) for one-group randomization  based on  10,000 replicates. Panels (A)-(F) correspond to $\rho_c$=0.001, 0.02, 0.06, 0.1, 0.3, and 0.5, respectively. } 
    \label{fig:centered}
\end{figure}

\begin{figure}[H]
    \centering 
  \includegraphics[width=0.9\textwidth]{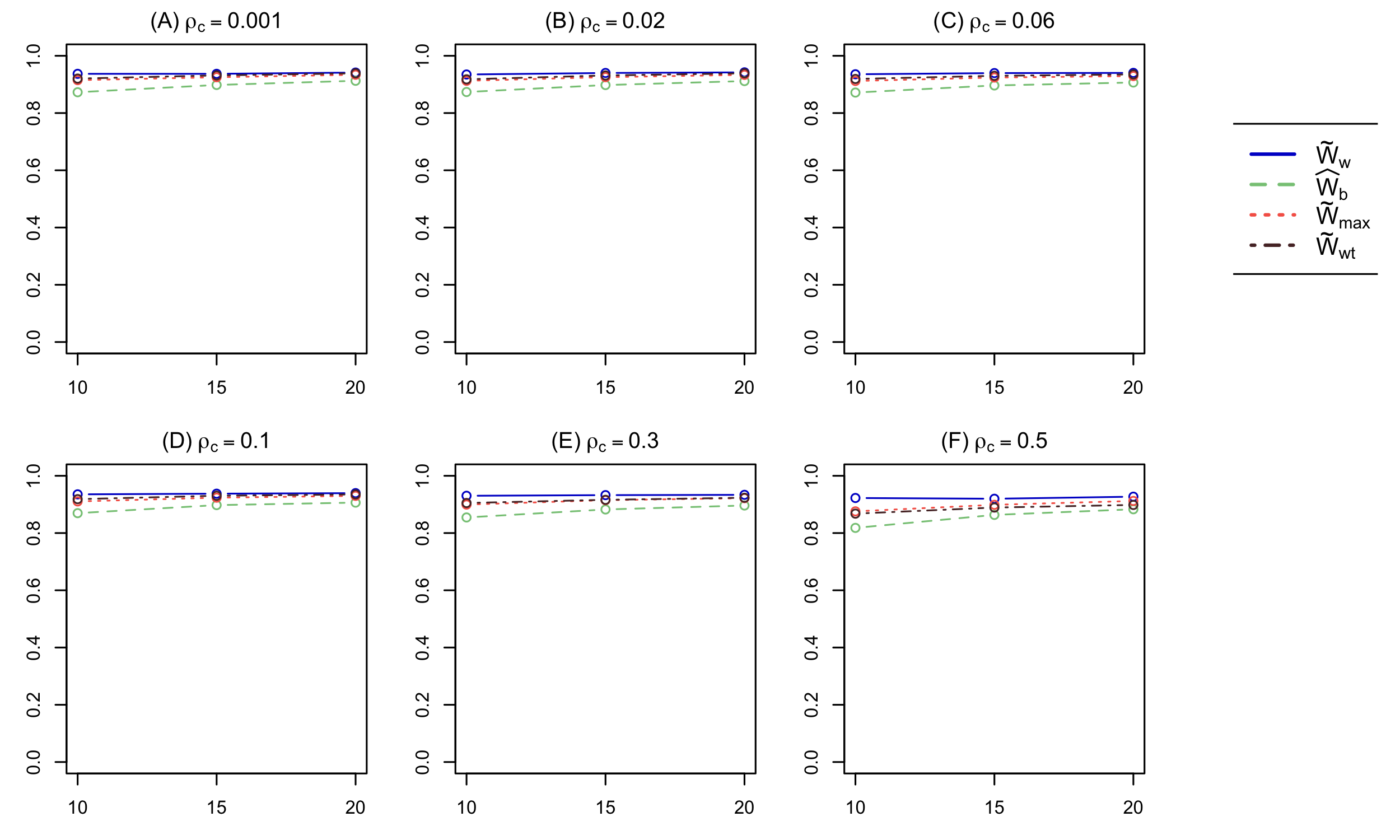}
    \caption{Powers for the tests based on $\widetilde{W}_w$ with Type 2 method, $\widehat{W}_b$, $\widetilde{W}_{max}$ and $\widetilde{W}_{wt}$ with small-sample corrections or adjustments for testing $H_0: D_w=0.5$, $H_0: D_b=0.5$, $H_0: D_b=D_w=0.5$, and $H_0: D_b=D_w=0.5$, respectively,  under the small-sample size scenarios (1. $n$=10, $m$=60; 2. $n$=15, $m$=40; 3. $n$=20, $m$=30) for two-group randomization based on 10,000 replicates. Panels (A)-(F) correspond to $\rho_c$=0.001, 0.02, 0.06, 0.1, 0.3, and 0.5, respectively.} 
    \label{fig:centered}
\end{figure}

\begin{figure}[H]
    \centering 
  \includegraphics[width=0.9\textwidth]{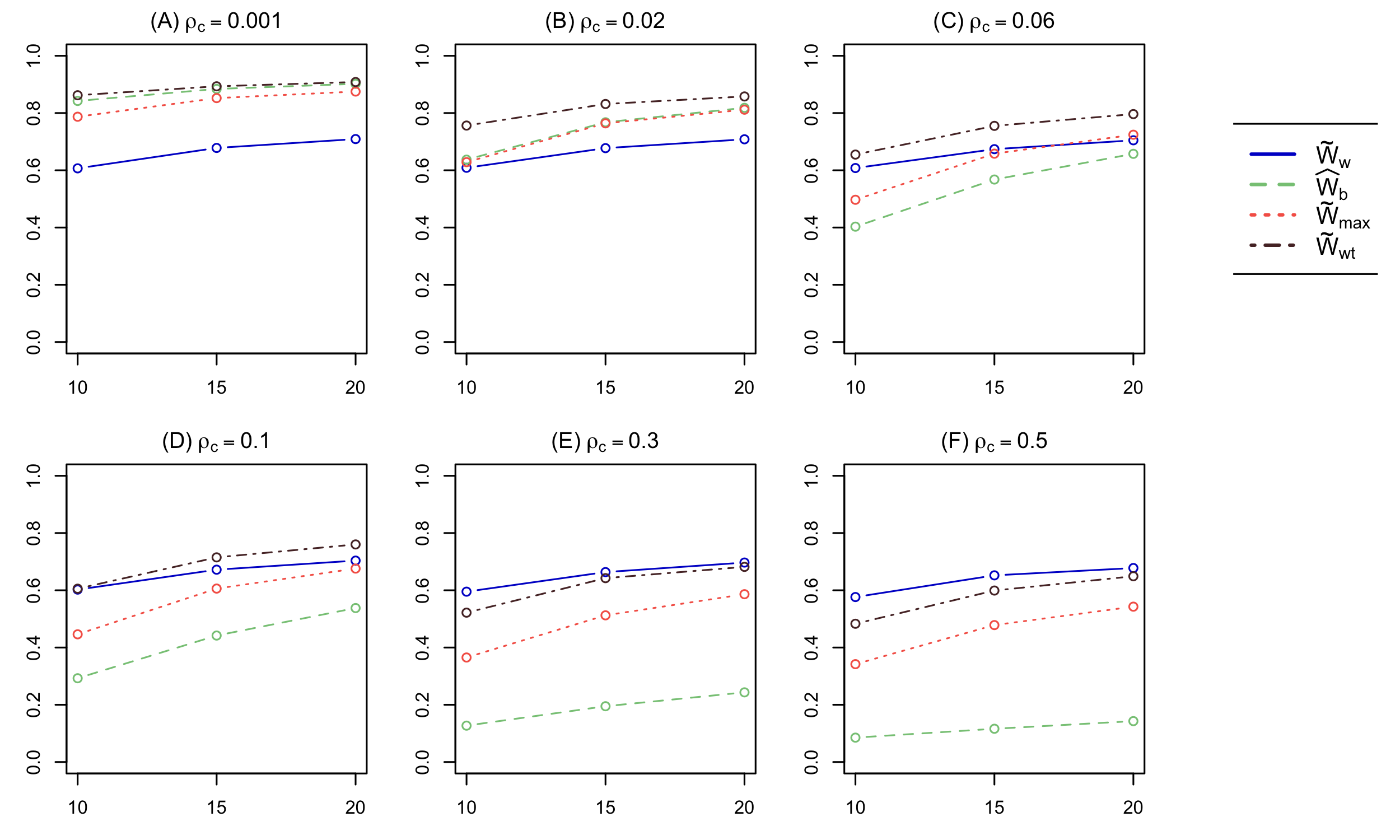}
    \caption{Powers for the tests based on $\widetilde{W}_w$ with Type 2 method, $\widehat{W}_b$, $\widetilde{W}_{max}$ and $\widetilde{W}_{wt}$ with small-sample corrections or adjustments for testing $H_0: D_w=0.5$, $H_0: D_b=0.5$, $H_0: D_b=D_w=0.5$, and $H_0: D_b=D_w=0.5$, respectively,  under the small-sample size scenarios (1. $n$=10, $m$=60; 2. $n$=15, $m$=40; 3. $n$=20, $m$=30) for mixture randomization based on 10,000 replicates. Panels (A)-(F) correspond to $\rho_c$=0.001, 0.02, 0.06, 0.1, 0.3, and 0.5, respectively. } 
    \label{fig:centered}
\end{figure}

\end{document}

%% file: main_CP_tables/Large_100_8_H1_b1.0_20260417.tex
\begin{table}[H]
\centering
\begin{tabular}{lllll}
  \hline
{Estimator or $\rho_c$}  & Bias & SE & SEE & CP \\ 
  \hline
  {One-group $\widehat{D}_b$} & & & & \\
  \hline
0.001 & -0.0003 & 0.0200 & 0.0199 & 0.950 \\ 
  0.02 & -0.0002 & 0.0211 & 0.0211 & 0.949 \\ 
  0.06 & -0.0001 & 0.0233 & 0.0233 & 0.950 \\ 
  0.1 & -0.0000 & 0.0253 & 0.0253 & 0.948 \\ 
  0.3 & 0.0001 & 0.0337 & 0.0337 & 0.948 \\ 
  0.5 & 0.0003 & 0.0404 & 0.0404 & 0.948 \\ 
    \hline
  {Two-group $\rho_c$=0.1}& & & & \\
  \hline
  $\widetilde{D}_{w}$  & 0.0001 & 0.0204 & 0.0201 & 0.946 \\ 
  $\widehat{D}_b$ & 0.0001 & 0.0190 & 0.0188 & 0.949 \\ 
   \hline
  {Mixture $\rho_c$=0.1} & & & & \\
   \hline
  $\widetilde{D}_{w}$  & 0.0000 & 0.0287 & 0.0283 & 0.942 \\ 
  $\widehat{D}_b$ & 0.0001 & 0.0223 & 0.0224 & 0.953 \\ 
   \hline
\end{tabular}
\caption{Simulation results for the large sample size scenario with 100 clusters and 8 subjects in each cluster based on 10,000 replicates. Bias is the average difference between parameter estimates and the true parameter value. SE is the sample standard deviation of the estimates. SEE is the average of the standard error estimates. CP is the coverage probability of the 95\% confidence interval.}
\end{table}

%% file: main_CP_tables/Small_new_10_60_CP_H1_b1.0_20260417.tex
\begin{table}[H]
\centering
\begin{tabular}{lllll}
  \hline
{Estimator or $\rho_c$} & Bias & SE & SEE & CP \\ 
  \hline
{One-group $\widehat{D}_b$ (w/ corrections)}& & & & \\
  \hline
0.001 & 0.0001 & 0.0233 & 0.0240 & 0.957 \\ 
  0.02 & 0.0005 & 0.0329 & 0.0336 & 0.954 \\ 
  0.06 & 0.0009 & 0.0469 & 0.0479 & 0.953 \\ 
  0.1 & 0.0011 & 0.0577 & 0.0589 & 0.952 \\ 
  0.3 & 0.0019 & 0.0947 & 0.0974 & 0.945 \\ 
  0.5 & 0.0026 & 0.1215 & 0.1257 & 0.936 \\ 
    \hline
  {Two-group $\rho_c$=0.1}& & & & \\
  \hline
  {$\widehat{D}_{w}$} & 0.0011 & 0.0241 & 0.0225 & 0.928 \\ 
{$\widetilde{D}_{w}$ (w/ Type 2 method)}& -0.0003 & 0.0229 & 0.0226 & 0.944 \\ 
{$\widetilde{D}_{w}$ (w/ Type 3 method)}& -0.0003 & 0.0229 & 0.0236 & 0.963 \\ 
 {$\widehat{D}_{b}$ (w/ corrections)} & -0.0003 & 0.0218 & 0.0212 & 0.952 \\ 
  \hline
  {Mixture $\rho_c$=0.1}& & & & \\
  \hline
{$\widehat{D}_{w}$} & 0.0009 & 0.0377 & 0.0356 & 0.931 \\ 
 {$\widetilde{D}_{w}$ (w/ Type 2 method)}& -0.0002 & 0.0362 & 0.0357 & 0.943 \\ 
 {$\widetilde{D}_{w}$ (w/ Type 3 method)}& -0.0002 & 0.0362 & 0.0373 & 0.946 \\ 
 {$\widehat{D}_{b}$ (w/ corrections)} & 0.0001 & 0.0476 & 0.0489 & 0.965 \\ 
   \hline
\end{tabular}
\caption{Simulation results for the small sample size scenario with 10 clusters and 60 subjects in each cluster based on 10,000 replicates. Bias is the average difference between parameter estimates and the true parameter value. SE is the sample standard deviation of the estimates. SEE is the average of the standard error estimates. CP is the coverage probability of the 95\% confidence interval.}
\end{table}

%% file: Suggested_Methods_v2.tex
\begin{table}[H]
\centering
\footnotesize  
\resizebox{\textwidth}{!}{
\begin{tabular}{|p{2.5cm}|p{3cm}|p{2.5cm}|p{5cm}|}
\hline
 \textbf{$n$ } &\textbf{$m$}  &\textbf{Within-cluster\newline estimator} &\textbf{Variance Estimator}\\ \hline

 Small 
& {Large\newline and small within-cluster correlation} 
&$\widehat{D}_{w}$ 
& Estimated by $1/{\sum_{i=1}^n m_i \widehat{\sigma}_i^{-2}}$
\\ \hline

Small 
& Large 
& $\widetilde{D}_{w}$ 
& [Type 2 method] Estimated by $\sum_{i=1}^n \widetilde{w}_i^2\sigma_i^2$ 
\\ \hline

Small 
& Moderate/Large
& $\widetilde{D}_{w}$ 
& [Type 3 method] Use influence functions with small sample size corrections
\\ \hline

 Large
& Small/Large 
& $\widetilde{D}_{w}$ 
& [Type 1 method]\newline Use influence functions, can estimated using the empirical version
$\sum_{i=1}^n \widetilde{w}_i^2 (\widehat{D}_{wi} - \widetilde{D}_{w})^2$ 
\\ \hline

 &  &\textbf{Between-cluster\newline estimator}  &  \\ \hline

Small
& Moderate/Large
& $\widehat{D}_b$ 
& Use influence functions with small sample size corrections
  \\ \hline

Large ($n\geq 15$)
& Small/Large 
& $\widehat{D}_b$ 
& Use influence function
\\ \hline

\end{tabular}
}
\caption{A guidance for the choice of the proposed methods in different practical settings.}
\label{Suggested_Methods_v2}
\end{table}

%% file: supp_CP_tables/Large_100_8_H1_b1.0_corr_20260417.tex
\begin{table}[H]
\centering
\begin{tabular}{llllll}
  \hline
Estimator & $\rho_c$ & Bias & SE & SEE & CP  \\ 
  \hline
{Two-group}& & & & \\
  \hline
  $\widetilde{D}_w$ & 0.001 & 0.0001 & 0.0205 & 0.0202 & 0.944  \\ 
  $\widehat{D}_b$ & 0.001  & 0.0001 & 0.0199 & 0.0197 & 0.950 \\ 
      \hline
  $\widetilde{D}_w$ & 0.02 & 0.0001 & 0.0205 & 0.0202 & 0.945  \\ 
  $\widehat{D}_b$ & 0.02 & 0.0001 & 0.0197 & 0.0196 & 0.949  \\ 
      \hline
  $\widetilde{D}_w$ & 0.06 & 0.0001 & 0.0204 & 0.0201 & 0.947  \\ 
  $\widehat{D}_b$ & 0.06 & 0.0001 & 0.0194 & 0.0192 & 0.949  \\ 
      \hline
  $\widetilde{D}_w$ & 0.3 & 0.0001 & 0.0201 & 0.0198 & 0.943  \\ 
  $\widehat{D}_b$ & 0.3 & 0.0001 & 0.0171 & 0.0169 & 0.947  \\ 
      \hline
  $\widetilde{D}_w$ & 0.5 & 0.0001 & 0.0195 & 0.0192 & 0.944  \\ 
  $\widehat{D}_b$ & 0.5 & 0.0001 & 0.0148 & 0.0146 & 0.950  \\ 
     \hline
   {Mixture}& & & & \\
   \hline
  $\widetilde{D}_w$ & 0.001 & -0.0001 & 0.0288 & 0.0284 & 0.942  \\ 
  $\widehat{D}_b$ & 0.001 & 0.0001 & 0.0198 & 0.0199 & 0.952  \\ 
      \hline
  $\widetilde{D}_w$ & 0.02 & 0.0000 & 0.0288 & 0.0284 & 0.943  \\ 
  $\widehat{D}_b$ & 0.02 & 0.0001 & 0.0203 & 0.0204 & 0.950  \\ 
      \hline
  $\widetilde{D}_w$ & 0.06 & -0.0001 & 0.0288 & 0.0283 & 0.940  \\ 
  $\widehat{D}_b$ & 0.06 & 0.0000 & 0.0214 & 0.0214 & 0.951  \\ 
      \hline
  $\widetilde{D}_w$ & 0.3 & -0.0001 & 0.0282 & 0.0278 & 0.942  \\ 
  $\widehat{D}_b$ & 0.3 & -0.0001 & 0.0267 & 0.0267 & 0.951  \\ 
      \hline
  $\widetilde{D}_w$ & 0.5 & -0.0001 & 0.0274 & 0.0270 & 0.941  \\ 
  $\widehat{D}_b$ & 0.5 & -0.0001 & 0.0307 & 0.0305 & 0.949  \\ 
   \hline
\end{tabular}
\caption{Simulation results for the large sample size scenario with 100 clusters and 8 subjects in each cluster based on 10,000 replicates. Bias is the average difference between parameter estimates and the true parameter value. SE is the sample standard deviation of the estimates. SEE is the average of the standard error estimates. CP is the coverage probability of the 95\% confidence interval. $\rho_c$=0.001, 0.02, 0.06, 0.3, 0.5.}
\end{table}

%% file: supp_CP_tables/Large_50_16_H1_b1.0_20260417.tex
\begin{table}[H]
\centering
\begin{tabular}{lllll}
  \hline
{Estimator or $\rho_c$} & Bias & SE & SEE & CP \\ 
  \hline
  {One-group $\widehat{D}_b$}& & & & \\
  \hline
0.001 & -0.0002 & 0.0199 & 0.0200 & 0.952 \\ 
  0.020 & -0.0003 & 0.0223 & 0.0223 & 0.953 \\ 
  0.060 & -0.0005 & 0.0267 & 0.0265 & 0.949 \\ 
  0.100 & -0.0006 & 0.0305 & 0.0302 & 0.949 \\ 
  0.300 & -0.0009 & 0.0449 & 0.0442 & 0.944 \\ 
  0.500 & -0.0012 & 0.0559 & 0.0549 & 0.944 \\ 
\hline
  {Two-group $\rho_c$=0.1}& & & & \\
  \hline
  $\widetilde{D}_w$ & -0.0001 & 0.0201 & 0.0197 & 0.940 \\ 
  $\widehat{D}_b$ & -0.0002 & 0.0190 & 0.0188 & 0.949 \\ 
\hline
  {Mixture $\rho_c$=0.1}& & & & \\
  \hline
  $\widetilde{D}_w$ & 0.0002 & 0.0288 & 0.0280 & 0.935 \\ 
  $\widehat{D}_b$ & 0.0004 & 0.0257 & 0.0254 & 0.951 \\ 
   \hline
\end{tabular}
\caption{Simulation results for the large sample size scenario with 50 clusters and 16 subjects in each cluster based on 10,000 replicates. Bias is the average difference between parameter estimates and the true parameter value. SE is the sample standard deviation of the estimates. SEE is the average of the standard error estimates. CP is the coverage probability of the 95\% confidence interval.}
\end{table}

%% file: supp_CP_tables/Large_200_4_H1_b1.0_20260417.tex
\begin{table}[H]
\centering
\begin{tabular}{lllll}
  \hline
{Estimator or $\rho_c$} & Bias & SE & SEE & CP \\ 
  \hline
    {One-group $\widehat{D}_b$}& & & & \\
  \hline
0.001 & -0.0001 & 0.0197 & 0.0199 & 0.955 \\ 
  0.020 & -0.0000 & 0.0202 & 0.0204 & 0.952 \\ 
  0.060 & 0.0000 & 0.0212 & 0.0214 & 0.954 \\ 
  0.100 & 0.0001 & 0.0222 & 0.0223 & 0.953 \\ 
  0.300 & 0.0002 & 0.0266 & 0.0267 & 0.950 \\ 
  0.500 & 0.0002 & 0.0304 & 0.0305 & 0.949 \\ 
\hline
    {Two-group $\rho_c$=0.1}& & & & \\
\hline
  $\widetilde{D}_w$ & -0.0001 & 0.0209 & 0.0207 & 0.944 \\ 
  $\widehat{D}_b$ & -0.0001 & 0.0190 & 0.0189 & 0.947 \\ 
\hline
    {Mixture $\rho_c$=0.1}& & & & \\
\hline
  $\widetilde{D}_w$ & 0.0004 & 0.0297 & 0.0292 & 0.945 \\ 
  $\widehat{D}_b$& 0.0005 & 0.0206 & 0.0207 & 0.951 \\ 
   \hline
\end{tabular}
\caption{Simulation results for the large sample size scenario with 200 clusters and 4 subjects in each cluster based on 10,000 replicates. Bias is the average difference between parameter estimates and the true parameter value. SE is the sample standard deviation of the estimates. SEE is the average of the standard error estimates. CP is the coverage probability of the 95\% confidence interval.}
\end{table}

%% file: supp_CP_tables/Small_old_10_60_CP_H1_b1.0_20260417.tex
\begin{table}[H]
\centering
\begin{tabular}{lllll}
  \hline
{Estimator or $\rho_c$} & Bias & SE & SEE & CP \\ 
  \hline
{One-group $\widehat{D}_b$}& & & & \\
  \hline
0.001 & 0.0001 & 0.0233 & 0.0232 & 0.947 \\ 
  0.02 & 0.0005 & 0.0329 & 0.0315 & 0.943 \\ 
  0.06 & 0.0009 & 0.0469 & 0.0441 & 0.942 \\ 
  0.1 & 0.0011 & 0.0577 & 0.0539 & 0.942 \\ 
  0.3 & 0.0019 & 0.0947 & 0.0885 & 0.933 \\ 
  0.5 & 0.0026 & 0.1215 & 0.1141 & 0.926 \\ 
    \hline
  {Two-group $\rho_c$=0.1}& & & & \\
  \hline
{$\widehat{D}_w$} & 0.0011 & 0.0241 & 0.0225 & 0.928 \\ 
  {$\widetilde{D}_w$ (w/ Type 1 method)}& -0.0003 & 0.0229 & 0.0211 & 0.907 \\ 
  {$\widehat{D}_b$} & -0.0003 & 0.0218 & 0.0207 & 0.934 \\ 
    \hline
  {Mixture $\rho_c$=0.1}& & & & \\
  \hline
{$\widehat{D}_w$}& 0.0009 & 0.0377 & 0.0356 & 0.931 \\ 
  {$\widetilde{D}_w$ (w/ Type 1 method)} & -0.0002 & 0.0362 & 0.0289 & 0.808 \\ 
  {$\widehat{D}_b$}& 0.0001 & 0.0476 & 0.0448 & 0.953 \\ 
   \hline
\end{tabular}
\caption{Simulation results for the small sample size scenario for methods without small-sample corrections, with 10 clusters and 60 subjects in each cluster based on 10,000 replicates. Bias is the average difference between parameter estimates and the true parameter value. SE is the sample standard deviation of the estimates. SEE is the average of the standard error estimates. CP is the coverage probability of the 95\% confidence interval.}
\end{table}

%% file: supp_CP_tables/Small_old_15_40_CP_H1_b1.0_20260417.tex
\begin{table}[H]
\centering
\begin{tabular}{lllll}
  \hline
{Estimator or $\rho_c$} & Bias & SE & SEE & CP \\ 
  \hline
{One-group $\widehat{D}_b$}& & & & \\
  \hline
0.001 & -0.0000 & 0.0234 & 0.0232 & 0.950 \\ 
  0.02& -0.0001 & 0.0298 & 0.0291 & 0.950 \\ 
  0.06 & -0.0004 & 0.0403 & 0.0389 & 0.949 \\ 
  0.1 & -0.0004 & 0.0487 & 0.0467 & 0.946 \\ 
  0.3 & -0.0006 & 0.0781 & 0.0748 & 0.941 \\ 
  0.5 & -0.0009 & 0.0996 & 0.0957 & 0.936 \\ 
      \hline
  {Two-group $\rho_c$=0.1}& & & & \\
  \hline
  {$\widehat{D}_w$}  & 0.0020 & 0.0250 & 0.0224 & 0.916 \\ 
  {$\widetilde{D}_w$ (w/ Type 1 method)} & -0.0002 & 0.0228 & 0.0218 & 0.921 \\ 
  {$\widehat{D}_b$}& -0.0002 & 0.0218 & 0.0211 & 0.943 \\ 
      \hline
  {Mixture $\rho_c$=0.1}& & & & \\
  \hline
  {$\widehat{D}_w$}  & 0.0018 & 0.0361 & 0.0328 & 0.917 \\ 
  {$\widetilde{D}_w$ (w/ Type 1 method)}& -0.0002 & 0.0333 & 0.0298 & 0.885 \\ 
  {$\widehat{D}_b$}& 0.0001 & 0.0399 & 0.0376 & 0.945 \\ 
   \hline
\end{tabular}
\caption{Simulation results for the small sample size scenario for methods without small-sample corrections, with 15 clusters and 40 subjects in each cluster based on 10,000 replicates. Bias is the average difference between parameter estimates and the true parameter value. SE is the sample standard deviation of the estimates. SEE is the average of the standard error estimates. CP is the coverage probability of the 95\% confidence interval.}
\end{table}

%% file: supp_CP_tables/Small_old_20_30_CP_H1_b1.0_20260417.tex
\begin{table}[H]
\centering
\begin{tabular}{lllll}
  \hline
{Estimator or $\rho_c$} & Bias & SE & SEE & CP \\ 
  \hline
{One-group $\widehat{D}_b$}& & & & \\
  \hline
0.001 & -0.0004 & 0.0230 & 0.0231 & 0.952 \\ 
  0.02& -0.0006 & 0.0280 & 0.0278 & 0.949 \\ 
  0.06 & -0.0008 & 0.0363 & 0.0358 & 0.947 \\ 
  0.1 & -0.0010 & 0.0431 & 0.0423 & 0.948 \\ 
  0.3 & -0.0013 & 0.0676 & 0.0662 & 0.943 \\ 
  0.5 & -0.0015 & 0.0858 & 0.0841 & 0.941 \\ 
    \hline
  {Two-group $\rho_c$=0.1}& & & & \\
  \hline
  {$\widehat{D}_w$}   & 0.0035 & 0.0265 & 0.0222 & 0.896 \\ 
  {$\widetilde{D}_w$ (w/ Type 1 method)} & 0.0002 & 0.0230 & 0.0221 & 0.928 \\ 
  {$\widehat{D}_b$}& 0.0001 & 0.0219 & 0.0213 & 0.947 \\ 
    \hline
  {Mixture $\rho_c$=0.1}& & & & \\
  \hline
  {$\widehat{D}_w$}   & 0.0028 & 0.0369 & 0.0314 & 0.902 \\ 
  {$\widetilde{D}_w$ (w/ Type 1 method)} & -0.0002 & 0.0327 & 0.0298 & 0.899 \\ 
  {$\widehat{D}_b$} & -0.0004 & 0.0354 & 0.0339 & 0.946 \\ 
   \hline
\end{tabular}
  \caption{Simulation results for the small sample size scenario for methods without small-sample corrections, with 20 clusters and 30 subjects in each cluster based on 10,000 replicates. Bias is the average difference between parameter estimates and the true parameter value. SE is the sample standard deviation of the estimates. SEE is the average of the standard error estimates. CP is the coverage probability of the 95\% confidence interval.}
\end{table}

%% file: supp_CP_tables/Small_new_15_40_CP_H1_b1.0_20260417.tex
\begin{table}[H]
\centering
\begin{tabular}{lllll}
  \hline
{Estimator or $\rho_c$} & Bias & SE & SEE & CP \\ 
  \hline
{One-group $\widehat{D}_b $ (w/ corrections)}& & & & \\
  \hline
0.001 & -0.0000 & 0.0234 & 0.0237 & 0.955 \\ 
  0.020 & -0.0001 & 0.0298 & 0.0303 & 0.956 \\ 
  0.060 & -0.0004 & 0.0403 & 0.0410 & 0.956 \\ 
  0.100 & -0.0004 & 0.0487 & 0.0495 & 0.955 \\ 
  0.300 & -0.0006 & 0.0781 & 0.0797 & 0.948 \\ 
  0.500 & -0.0009 & 0.0996 & 0.1021 & 0.945 \\ 
      \hline
  {Two-group $\rho_c$=0.1}& & & & \\
  \hline
 {$\widehat{D}_{w}$}  & 0.0020 & 0.0250 & 0.0224 & 0.916 \\ 
{$\widetilde{D}_{w}$ (w/ Type 2 method)} & -0.0002 & 0.0228 & 0.0224 & 0.944 \\ 
{$\widetilde{D}_{w}$ (w/ Type 3 method)}& -0.0002 & 0.0228 & 0.0234 & 0.957 \\ 
 {$\widehat{D}_{b}$ (w/ corrections)} & -0.0002 & 0.0218 & 0.0215 & 0.949 \\ 
    \hline
  {Mixture $\rho_c$=0.1}& & & & \\
  \hline
 {$\widehat{D}_{w}$}  & 0.0018 & 0.0361 & 0.0328 & 0.917 \\ 
{$\widetilde{D}_{w}$ (w/ Type 2 method)}  & -0.0002 & 0.0333 & 0.0329 & 0.943  \\ 
{$\widetilde{D}_{w}$ (w/ Type 3 method)} & -0.0002 & 0.0333 & 0.0348 & 0.958 \\ 
 {$\widehat{D}_{b}$ (w/ corrections)} & 0.0001 & 0.0399 & 0.0396 & 0.957 \\ 
   \hline
\end{tabular}
\caption{Simulation results for the small sample size scenario for methods with small-sample corrections or adjustments, with 15 clusters and 40 subjects in each cluster based on 10,000 replicates. Bias is the average difference between parameter estimates and the true parameter value. SE is the sample standard deviation of the estimates. SEE is the average of the standard error estimates. CP is the coverage probability of the 95\% confidence interval.}
\end{table}

%% file: supp_CP_tables/Small_new_20_30_CP_H1_b1.0_20260417.tex
\begin{table}[H]
\centering
\begin{tabular}{lllll}
  \hline
{Estimator or $\rho_c$} & Bias & SE & SEE & CP \\ 
  \hline
{One-group $\widehat{D}_b$ (w/ corrections)}& & & & \\
  \hline
0.001 & -0.0004 & 0.0230 & 0.0235 & 0.957 \\ 
  0.02 & -0.0006 & 0.0280 & 0.0286 & 0.955 \\ 
  0.06 & -0.0008 & 0.0363 & 0.0372 & 0.953 \\ 
  0.1 & -0.0010 & 0.0431 & 0.0441 & 0.953 \\ 
  0.3 & -0.0013 & 0.0676 & 0.0695 & 0.949 \\ 
  0.5 & -0.0015 & 0.0858 & 0.0884 & 0.949 \\ 
    \hline
  {Two-group $\rho_c$=0.1}& & & & \\
  \hline
  {$\widehat{D}_{w}$}  & 0.0035 & 0.0265 & 0.0222 & 0.896 \\ 
{$\widetilde{D}_{w}$ (w/ Type 2 method)} & 0.0002 & 0.0230 & 0.0223 & 0.943 \\ 
{$\widetilde{D}_{w}$ (w/ Type 3 method)} & 0.0002 & 0.0230 & 0.0233 & 0.956 \\ 
 {$\widehat{D}_{b}$ (w/ corrections)}  & 0.0001 & 0.0219 & 0.0216 & 0.951 \\ 
    \hline
  {Mixture $\rho_c$=0.1}& & & & \\
  \hline
  {$\widehat{D}_{w}$}  & 0.0028 & 0.0369 & 0.0314 & 0.902 \\ 
{$\widetilde{D}_{w}$ (w/ Type 2 method)} & -0.0002 & 0.0327 & 0.0316 & 0.939 \\ 
{$\widetilde{D}_{w}$ (w/ Type 3 method)} & -0.0002 & 0.0327 & 0.0334 & 0.954 \\ 
 {$\widehat{D}_{b}$ (w/ corrections)} & -0.0004 & 0.0354 & 0.0351 & 0.952 \\ 
   \hline
\end{tabular}
\caption{Simulation results for the small sample size scenario for methods with small-sample corrections or adjustments, with 20 clusters and 30 subjects in each cluster based on 10,000 replicates. Bias is the average difference between parameter estimates and the true parameter value. SE is the sample standard deviation of the estimates. SEE is the average of the standard error estimates. CP is the coverage probability of the 95\% confidence interval.}
\end{table}

%% file: Reference.bib
@article{tukey1953problem,
  title={The problem of multiple comparisons},
  author={Tukey, John Wilder},
  journal={Multiple Comparisons},
  year={1953},
  publisher={Chapman and Hall}
}

@book{jennison2000group,
author={Jennison, C and Turnbull, B. W.},
title={Group Sequential Methods with Applications to Clinical Trials},
address={Boca Raton, FL},
publisher={Chapman \& Hall/CRC}, 
year={2000},
}

@article{neuhaus1998between,
  title={Between-and within-cluster covariate effects in the analysis of clustered data},
  author={Neuhaus, John M and Kalbfleisch, Jack D},
  journal={Biometrics},
  pages={638--645},
  year={1998},
  publisher={JSTOR}
}

@article{roberts2005design,
  title={Design and analysis of clinical trials with clustering effects due to treatment},
  author={Roberts, Chris and Roberts, Stephen A},
  journal={Clinical Trials},
  volume={2},
  number={2},
  pages={152--162},
  year={2005},
  publisher={Sage Publications Sage CA: Thousand Oaks, CA}
}

@article{lyu2020comparison,
  title={Comparison of perioperative parameters in femtosecond laser-assisted cataract surgery using 3 nuclear fragmentation patterns},
  author={Lyu, Danni and Shen, Zeren and Zhang, Lifang and Qin, Zhenwei and Ni, Shuang and Wang, Wei and Zhu, Yanan and Yao, Ke},
  journal={American Journal of Ophthalmology},
  volume={213},
  pages={283--292},
  year={2020},
  publisher={Elsevier}
}

@article{price2021randomized,
  title={Randomized, double-masked trial of netarsudil 0.02\% ophthalmic solution for prevention of corticosteroid-induced ocular hypertension},
  author={Price, Marianne O and Feng, Matthew T and Price Jr, Francis W},
  journal={American Journal of Ophthalmology},
  volume={222},
  pages={382--387},
  year={2021},
  publisher={Elsevier}
}

@article{pall2019management,
  title={Management of ocular allergy itch with an antihistamine-releasing contact lens},
  author={Pall, Brian and Gomes, Paul and Yi, Frank and Torkildsen, Gail},
  journal={Cornea},
  volume={38},
  number={6},
  pages={713--717},
  year={2019},
  publisher={LWW}
}

@book{hayesClusterRandomisedTrials2017,
  title = {Cluster {{Randomised Trials}}},
  author = {Hayes, Richard J. and Moulton, Lawrence H.},
  year = {2017},
  month = jul,
  edition = {2},
  publisher = {{Chapman and Hall/CRC}},
  address = {New York},
  doi = {10.4324/9781315370286},
  isbn = {978-1-315-37028-6}
}

@article{ahnEvaluationWeightedChiSquare2003,
  title = {An {{Evaluation}} of {{Weighted Chi-Square Statistics}} for {{Clustered Binary Data}}},
  author = {Ahn, Chul and Jung, Sin-Ho and Kang, Seung-Ho},
  year = {2003},
  month = jan,
  journal = {Drug Information Journal},
  volume = {37},
  number = {1},
  pages = {91--99},
  publisher = {SAGE Publications},
  issn = {0092-8615},
  doi = {10.1177/009286150303700111},
  urldate = {2025-08-25},
  langid = {english}
}

@article{dattaRankSumTestsClustered2005,
  title = {Rank-{{Sum Tests}} for {{Clustered Data}}},
  author = {Datta, Somnath and Satten, Glen A},
  year = {2005},
  month = sep,
  journal = {Journal of the American Statistical Association},
  volume = {100},
  number = {471},
  pages = {908--915},
  publisher = {ASA Website},
  issn = {0162-1459},
  doi = {10.1198/016214504000001583},
  urldate = {2025-08-25}
}

@article{dattaSignedRankTestClustered2008,
  title = {A {{Signed-Rank Test}} for {{Clustered Data}}},
  author = {Datta, Somnath and Satten, Glen A.},
  year = {2008},
  month = jun,
  journal = {Biometrics},
  volume = {64},
  number = {2},
  pages = {501--507},
  issn = {0006-341X},
  doi = {10.1111/j.1541-0420.2007.00923.x},
  urldate = {2025-08-25}
}

@article{donnerAnalysisSitespecificData1988,
  title = {Analysis of {{Site-specific Data}} in {{Dental Studies}}},
  author = {Donner, A. and Banting, D.},
  year = {1988},
  month = nov,
  journal = {Journal of Dental Research},
  volume = {67},
  number = {11},
  pages = {1392--1395},
  publisher = {SAGE Publications Inc},
  issn = {0022-0345},
  doi = {10.1177/00220345880670110601},
  urldate = {2025-08-25},
  langid = {english}
}

@article{donnerStatisticalMethodsOphthalmology1989,
  title = {Statistical {{Methods}} in {{Ophthalmology}}: {{An Adjusted Chi-Square Approach}}},
  shorttitle = {Statistical {{Methods}} in {{Ophthalmology}}},
  author = {Donner, Allan},
  year = {1989},
  journal = {Biometrics},
  volume = {45},
  number = {2},
  eprint = {2531501},
  eprinttype = {jstor},
  pages = {605--611},
  publisher = {[Wiley, International Biometric Society]},
  issn = {0006-341X},
  doi = {10.2307/2531501},
  urldate = {2025-08-25}
}

@article{duttaRanksumTestClustered2016,
  title = {A Rank-Sum Test for Clustered Data When the Number of Subjects in a Group within a Cluster Is Informative},
  author = {Dutta, Sandipan and Datta, Somnath},
  year = {2016},
  journal = {Biometrics},
  volume = {72},
  number = {2},
  pages = {432--440},
  issn = {1541-0420},
  doi = {10.1111/biom.12447},
  urldate = {2025-08-25},
  copyright = {{\copyright} 2015, The International Biometric Society},
  langid = {english}
}

@article{evansDesirabilityOutcomeRanking2015a,
  title = {Desirability of {{Outcome Ranking}} ({{DOOR}}) and {{Response Adjusted}} for {{Duration}} of {{Antibiotic Risk}} ({{RADAR}})},
  author = {Evans, Scott R. and Rubin, Daniel and Follmann, Dean and Pennello, Gene and Huskins, W. Charles and Powers, John H. and Schoenfeld, David and {Chuang-Stein}, Christy and Cosgrove, Sara E. and Fowler, Jr, Vance G. and Lautenbach, Ebbing and Chambers, Henry F.},
  year = {2015},
  month = sep,
  journal = {Clinical Infectious Diseases},
  volume = {61},
  number = {5},
  pages = {800--806},
  issn = {1058-4838},
  doi = {10.1093/cid/civ495},
  urldate = {2025-08-25}
}

@article{jungSampleSizeCalculations2001,
  title = {Sample Size Calculations for Clustered Binary Data},
  author = {Jung, Sin-Ho and Kang, Seung-Ho and Ahn, Chul},
  year = {2001},
  journal = {Statistics in Medicine},
  volume = {20},
  number = {13},
  pages = {1971--1982},
  issn = {1097-0258},
  doi = {10.1002/sim.846},
  urldate = {2025-08-25},
  copyright = {Copyright {\copyright} 2001 John Wiley \& Sons, Ltd.},
  langid = {english}
}

@article{kangSampleSizeCalculation2003,
  title = {Sample {{Size Calculation}} for {{Dichotomous Outcomes}} in {{Cluster Randomization Trials}} with {{Varying Cluster Size}}},
  author = {Kang, Seung Ho and Ahn, Chul W. and Jung, Sin Ho},
  year = {2003},
  journal = {Drug Information Journal},
  volume = {37},
  number = {1},
  pages = {109--114},
  issn = {0092-8615},
  urldate = {2025-08-25}
}

@article{katheriaUmbilicalCordMilking2023,
  title = {Umbilical Cord Milking in Nonvigorous Infants: A Cluster-Randomized Crossover Trial},
  shorttitle = {Umbilical Cord Milking in Nonvigorous Infants},
  author = {Katheria, Anup C. and Clark, Erin and Yoder, Bradley and Schm{\"o}lzer, Georg M. and Yan Law, Brenda Hiu and {El-Naggar}, Walid and Rittenberg, David and Sheth, Sheetal and Mohamed, Mohamed A. and Martin, Courtney and Vora, Farha and Lakshminrusimha, Satyan and Underwood, Mark and Mazela, Jan and Kaempf, Joseph and Tomlinson, Mark and Gollin, Yvonne and Fulford, Kevin and Goff, Yvonne and Wozniak, Paul and Baker, Katherine and Rich, Wade and Morales, Ana and Varner, Michael and Poeltler, Debra and Vaucher, Yvonne and Mercer, Judith and Finer, Neil and El Ghormli, Laure and Rice, Madeline Murguia},
  year = {2023},
  month = feb,
  journal = {American Journal of Obstetrics and Gynecology},
  volume = {228},
  number = {2},
  pages = {217.e1-217.e14},
  issn = {0002-9378},
  doi = {10.1016/j.ajog.2022.08.015},
  urldate = {2025-08-25}
}

@article{larocqueTwoSampleTests2010b,
  title = {Two Sample Tests for the Nonparametric {{Behrens}}--{{Fisher}} Problem with Clustered Data},
  author = {Larocque, Denis and Haataja, Riina and Nevalainen, Jaakko and Oja, Hannu},
  year = {2010},
  month = aug,
  journal = {Journal of Nonparametric Statistics},
  volume = {22},
  number = {6},
  pages = {755--771},
  publisher = {Taylor \& Francis},
  issn = {1048-5252},
  doi = {10.1080/10485250903469728},
  urldate = {2025-08-25}
}

@article{raoSimpleMethodAnalysis1992,
  title = {A {{Simple Method}} for the {{Analysis}} of {{Clustered Binary Data}}},
  author = {Rao, J. N. K. and Scott, A. J.},
  year = {1992},
  journal = {Biometrics},
  volume = {48},
  number = {2},
  eprint = {2532311},
  eprinttype = {jstor},
  pages = {577--585},
  publisher = {International Biometric Society},
  issn = {0006-341X},
  doi = {10.2307/2532311},
  urldate = {2025-08-25}
}

@article{rosnerExtensionRankSum2006,
  title = {Extension of the {{Rank Sum Test}} for {{Clustered Data}}: {{Two-Group Comparisons}} with {{Group Membership Defined}} at the {{Subunit Level}}},
  shorttitle = {Extension of the {{Rank Sum Test}} for {{Clustered Data}}},
  author = {Rosner, Bernard and Glynn, Robert J. and Lee, Mei-Ling T.},
  year = {2006},
  month = dec,
  journal = {Biometrics},
  volume = {62},
  number = {4},
  pages = {1251--1259},
  issn = {0006-341X},
  doi = {10.1111/j.1541-0420.2006.00582.x},
  urldate = {2025-08-25}
}

@article{rosnerIncorporationClusteringEffects2003a,
  title = {Incorporation of {{Clustering Effects}} for the {{Wilcoxon Rank Sum Test}}: {{A Large-Sample Approach}}},
  shorttitle = {Incorporation of {{Clustering Effects}} for the {{Wilcoxon Rank Sum Test}}},
  author = {Rosner, Bernard and Glynn, Robert J. and Ting Lee, Mei-Ling},
  year = {2003},
  month = dec,
  journal = {Biometrics},
  volume = {59},
  number = {4},
  pages = {1089--1098},
  issn = {0006-341X},
  doi = {10.1111/j.0006-341X.2003.00125.x},
  urldate = {2025-08-25}
}

@article{rosnerMultivariateMethodsClustered1997,
  title = {Multivariate {{Methods}} for {{Clustered Ordinal Data}} with {{Applications}} to {{Survival Analysis}}},
  author = {Rosner, Bernard and Glynn, Robert J.},
  year = {1997},
  journal = {Statistics in Medicine},
  volume = {16},
  number = {4},
  pages = {357--372},
  issn = {1097-0258},
  doi = {10.1002/(SICI)1097-0258(19970228)16:4<357::AID-SIM420>3.0.CO;2-3},
  urldate = {2025-08-25},
  copyright = {Copyright {\copyright} 1997 John Wiley \& Sons, Ltd.},
  langid = {english}
}

@article{rosnerUseMannWhitney1999,
  title = {Use of the {{Mann}}--{{Whitney U-test}} for Clustered Data},
  author = {Rosner, B. and Grove, D.},
  year = {1999},
  journal = {Statistics in Medicine},
  volume = {18},
  number = {11},
  pages = {1387--1400},
  issn = {1097-0258},
  doi = {10.1002/(SICI)1097-0258(19990615)18:11<1387::AID-SIM126>3.0.CO;2-V},
  urldate = {2025-08-25},
  copyright = {Copyright {\copyright} 1999 John Wiley \& Sons, Ltd.},
  langid = {english}
}

@article{rosnerWilcoxonSignedRank2006,
  title = {The {{Wilcoxon Signed Rank Test}} for {{Paired Comparisons}} of {{Clustered Data}}},
  author = {Rosner, Bernard and Glynn, Robert J. and Lee, Mei-Ling T.},
  year = {2006},
  month = mar,
  journal = {Biometrics},
  volume = {62},
  number = {1},
  pages = {185--192},
  issn = {0006-341X},
  doi = {10.1111/j.1541-0420.2005.00389.x},
  urldate = {2025-08-25}
}

@article{shuLongitudinalBenefitRisk2025,
  title = {Longitudinal {{Benefit}}:Risk {{Analysis}} through the {{Desirability}} of {{Outcome Ranking}} ({{DOOR}}) with {{Application}} to {{ACTT-1 Trial}}},
  shorttitle = {Longitudinal {{Benefit}}},
  author = {Shu, Shiyu and Diao, Guoqing and Hamasaki, Toshimitsu and Evans, Scott},
  year = {2025},
  month = jul,
  journal = {Statistics in Biopharmaceutical Research},
  volume = {17},
  number = {3},
  pages = {488--495},
  publisher = {ASA Website},
  issn = {null},
  doi = {10.1080/19466315.2024.2413059},
  urldate = {2025-08-25}
}

@incollection{hamasakiDesirabilityOutcomeRanking2025,
  title = {The {{Desirability}} of {{Outcome Ranking}} ({{DOOR}})},
  booktitle = {Handbook of {{Generalized Pairwise Comparisons}}},
  author = {Hamasaki, Toshimitsu and Evans, Scott R.},
  year = {2025},
  publisher = {{Chapman and Hall/CRC}}
}

@article{rosnerPowerSampleSize2011,
  title = {Power and {{Sample Size Estimation}} for the {{Clustered Wilcoxon Test}}},
  author = {Rosner, Bernard and Glynn, Robert J.},
  year = {2011},
  month = jun,
  journal = {Biometrics},
  volume = {67},
  number = {2},
  pages = {646--653},
  issn = {0006-341X},
  doi = {10.1111/j.1541-0420.2010.01488.x},
  urldate = {2025-09-12}
}

@article{Diao13112025,
author = {Guoqing Diao and Qiang Zhang and Toshimitsu Hamasaki and Scott Evans},
title = {Group Sequential Design and Monitoring of Clustered Data in Randomized Eye Trials},
journal = {Statistics in Biopharmaceutical Research},
volume = {0},
number = {0},
pages = {1--10},
year = {2025},
publisher = {ASA Website},
doi = {10.1080/19466315.2025.2566339},
URL = {   
        https://doi.org/10.1080/19466315.2025.2566339
},
eprint = { 
    
        https://doi.org/10.1080/19466315.2025.2566339
}

}

@article{hamasakiPatientcentricParadigmTool2025,
  title = {A Patient-Centric Paradigm and Tool for Clinical Research: The {{DOOR}} Is Open},
  shorttitle = {A Patient-Centric Paradigm and Tool for Clinical Research},
  author = {Hamasaki, Toshimitsu and He, Yijie and Wu, Qihang and {Howard-Anderson}, Jessica and Boucher, Helen W. and Doernberg, Sarah B. and Holland, Thomas L. and Powers, John H. and Wang, Jing and Diao, Guoqing and {van Duin}, David and Fowler, Vance G. and Chambers, Henry F. and Evans, Scott R.},
  year = 2025,
  month = nov,
  journal = {Antimicrobial Agents and Chemotherapy},
  volume = {0},
  number = {0},
  pages = {e01478-25},
  publisher = {American Society for Microbiology},
  doi = {10.1128/aac.01478-25},
  urldate = {2025-11-29}
}

@article{turnerDalbavancinTreatmentStaphylococcus2025,
  title = {Dalbavancin for {{Treatment}} of {{Staphylococcus}} Aureus {{Bacteremia}}: {{The DOTS Randomized Clinical Trial}}},
  shorttitle = {Dalbavancin for {{Treatment}} of {{Staphylococcus}} Aureus {{Bacteremia}}},
  author = {Turner, Nicholas A. and Hamasaki, Toshimitsu and Doernberg, Sarah B. and Lodise, Thomas P. and King, Heather A. and Ghazaryan, Varduhi and Cosgrove, Sara E. and Jenkins, Timothy C. and Liu, Catherine and Sharma, Shrabani and Zaharoff, Smitha and Wahid, Lana and Renard, Valerie J. and Cook, Paul and Raad, Issam and Hachem, Ray and Chaftari, Anne-Marie and Sims, Matthew and DeMarco, Carmen and Miller, Loren G. and McCarthy, Matthew W. and Morse, Caryn G. and Lucasti, Chris and Forrest, Graeme N. and Cherabuddi, Kartikeya and Polk, Christopher and Fazili, Tasaduq and Rupp, Mark E. and Thompson, III, George R. and Kim, Kami and Strnad, Luke and Schnee, Amanda E. and McKinnell, James A. and Ramesh, Mayur and Silveira, Fernanda P. and McCarty, Todd P. and Lee, Todd C. and McDonald, Emily G. and Paolino, Kristopher and Wiegand, Katie and Wall, Alison and Riccobene, Todd and Patel, Rinal and Rappo, Urania and Evans, Scott and Chambers, Henry F. and Fowler, Jr, Vance G. and Holland, Thomas L. and {Antibacterial Resistance Leadership Group}},
  year = 2025,
  month = sep,
  journal = {JAMA},
  volume = {334},
  number = {10},
  pages = {866--877},
  issn = {0098-7484},
  doi = {10.1001/jama.2025.12543},
  urldate = {2025-11-29}
}

@article{tammaClinicalImpactCeftriaxone2022,
  title = {Clinical {{Impact}} of {{Ceftriaxone Resistance}} in {{Escherichia}} Coli {{Bloodstream Infections}}: {{A Multicenter Prospective Cohort Study}}},
  shorttitle = {Clinical {{Impact}} of {{Ceftriaxone Resistance}} in {{Escherichia}} Coli {{Bloodstream Infections}}},
  author = {Tamma, Pranita D and Komarow, Lauren and Ge, Lizhao and {Garcia-Diaz}, Julia and Herc, Erica S and Doi, Yohei and Arias, Cesar A and Albin, Owen and Saade, Elie and Miller, Loren G and Jacob, Jesse T and Satlin, Michael J and Krsak, Martin and Huskins, W Charles and Dhar, Sorabh and Shelburne, Samuel A and Hill, Carol and Baum, Keri R and Bhojani, Minal and {Greenwood-Quaintance}, Kerryl E and {Schmidt-Malan}, Suzannah M and Patel, Robin and Evans, Scott R and Chambers, Henry F and Fowler, Jr, Vance G and {van Duin}, David and {for the Antibacterial Resistance Leadership Group}},
  year = 2022,
  month = nov,
  journal = {Open Forum Infectious Diseases},
  volume = {9},
  number = {11},
  pages = {ofac572},
  issn = {2328-8957},
  doi = {10.1093/ofid/ofac572},
  urldate = {2025-11-29}
}

@article{fangSampleSizeDetermination2025,
  title = {Sample Size Determination for Win Statistics in Cluster-Randomized Trials},
  author = {Fang, Xi and Cao, Zhiqiang and Li, Fan},
  year = 2025,
  month = oct,
  eprint = {2510.22709},
  archiveprefix = {arXiv},
  note={{\em arXiv}:2510.22709}
}

@article{spybrookProgressDecadeExamination2016,
  title = {Progress in the Past Decade: An Examination of the Precision of Cluster Randomized Trials Funded by the {{U}}.{{S}}. {{Institute}} of {{Education Sciences}}},
  shorttitle = {Progress in the Past Decade},
  author = {Spybrook, Jessaca and Shi, Ran and Kelcey, Benjamin},
  year = 2016,
  month = jul,
  journal = {International Journal of Research \& Method in Education},
  volume = {39},
  number = {3},
  pages = {255--267},
  publisher = {Routledge},
  issn = {1743-727X},
  doi = {10.1080/1743727X.2016.1150454},
  urldate = {2026-01-16}
}

@article{cookEmergentPrinciplesDesign2005,
  title = {Emergent {{Principles}} for the {{Design}}, {{Implementation}}, and {{Analysis}} of {{Cluster-Based Experiments}} in {{Social Science}}},
  author = {Cook, Thomas D.},
  year = 2005,
  month = may,
  journal = {The ANNALS of the American Academy of Political and Social Science},
  volume = {599},
  number = {1},
  pages = {176--198},
  publisher = {SAGE Publications Inc},
  issn = {0002-7162},
  doi = {10.1177/0002716205275738},
  urldate = {2026-01-16},
  langid = {english}
}

@article{cornfieldSYMPOSIUMCHDPREVENTION1978,
  title = {Randomization by group: a formal analysis},
  shorttitle = {{{SYMPOSIUM ON CHD PREVENTION TRIALS}}},
  author = {Cornfield, Jerome},
  year = 1978,
  month = aug,
  journal = {American Journal of Epidemiology},
  volume = {108},
  number = {2},
  pages = {100--102},
  issn = {0002-9262},
  doi = {10.1093/oxfordjournals.aje.a112592},
  urldate = {2026-01-20}
}

@article{donnerAspectsDesignAnalysis1998,
  title = {Some Aspects of the Design and Analysis of Cluster Randomization Trials},
  author = {Donner, Allan},
  year = 1998,
  journal = {Journal of the Royal Statistical Society: Series C (Applied Statistics)},
  volume = {47},
  number = {1},
  pages = {95--113},
  issn = {1467-9876},
  doi = {10.1111/1467-9876.00100},
  urldate = {2026-01-20},
  copyright = {1998 Royal Statistical Society},
  langid = {english}
}

@article{donnerRegressionApproachAnalysis,
  title = {A {{Regression Approach}} to the {{Analysis}} of {{Data Arising}} from {{Cluster Randomization}}},
  author = {Donner, Allan},
  year = 1985,
  month = jun,
  journal = {International Journal of Epidemiology},
  volume = {14},
  number = {2},
  pages = {322--326},
  issn = {0300-5771},
  doi = {10.1093/ije/14.2.322},
  urldate = {2026-03-13}
}

@article{hedekerRandomeffectsRegressionModels1994,
  title = {Random-Effects Regression Models for Clustered Data with an Example from Smoking Prevention Research},
  author = {Hedeker, Donald and Gibbons, Robert D. and Flay, Brian R.},
  year = 1994,
  journal = {Journal of Consulting and Clinical Psychology},
  volume = {62},
  number = {4},
  pages = {757--765},
  publisher = {American Psychological Association},
  address = {US},
  issn = {1939-2117},
  doi = {10.1037/0022-006X.62.4.757}
}

@article{lairdMaximumLikelihoodComputations1987,
  title = {Maximum {{Likelihood Computations}} with {{Repeated Measures}}: {{Application}} of the {{EM Algorithm}}},
  shorttitle = {Maximum {{Likelihood Computations}} with {{Repeated Measures}}},
  author = {Laird, Nan and Lange, Nicholas and Stram, Daniel},
  year = 1987,
  month = mar,
  journal = {Journal of the American Statistical Association},
  volume = {82},
  number = {397},
  pages = {97--105},
  publisher = {Taylor \& Francis},
  issn = {0162-1459},
  doi = {10.1080/01621459.1987.10478395},
  urldate = {2026-01-20}
}

@article{liangLongitudinalDataAnalysis1986,
  title = {Longitudinal Data Analysis Using Generalized Linear Models},
  author = {Liang, KUNG-YEE and Zeger, SCOTT L.},
  year = 1986,
  month = apr,
  journal = {Biometrika},
  volume = {73},
  number = {1},
  pages = {13--22},
  issn = {0006-3444},
  doi = {10.1093/biomet/73.1.13},
  urldate = {2026-01-20}
}

@article{chamberlainDesirabilityOutcomeRanking2023,
  title = {Desirability of {{Outcome Ranking}} for {{Status Epilepticus}}},
  author = {Chamberlain, James M. and Kapur, Jaideep and Silbergleit, Robert S. and Elm, Jordan J. and Rosenthal, Eric S. and Bleck, Thomas P. and Shinnar, Shlomo and Zetabchi, Shahriar and Evans, Scott R.},
  year = 2023,
  month = oct,
  journal = {Neurology},
  volume = {101},
  number = {16},
  pages = {e1633-e1639},
  publisher = {Wolters Kluwer},
  doi = {10.1212/WNL.0000000000207684},
  urldate = {2026-01-21}
}

@article{rosnerStatisticalMethodsOphthalmology1982,
  title = {Statistical {{Methods}} in {{Ophthalmology}}: {{An Adjustment}} for the {{Intraclass Correlation}} between {{Eyes}}},
  shorttitle = {Statistical {{Methods}} in {{Ophthalmology}}},
  author = {Rosner, Bernard},
  year = 1982,
  journal = {Biometrics},
  volume = {38},
  number = {1},
  eprint = {2530293},
  eprinttype = {jstor},
  pages = {105--114},
  publisher = {[Wiley, International Biometric Society]},
  issn = {0006-341X},
  doi = {10.2307/2530293},
  urldate = {2026-01-23}
}

@article{wangChoosingUnitRandomization2024,
  title = {Choosing the {{Unit}} of {{Randomization}} --- {{Individual}} or {{Cluster}}?},
  author = {Wang, Rui},
  year = 2024,
  month = mar,
  journal = {NEJM Evidence},
  volume = {3},
  number = {4},
  pages = {EVIDe2400037},
  publisher = {Massachusetts Medical Society},
  doi = {10.1056/EVIDe2400037},
  urldate = {2026-01-23}
}

@article{cookStatisticalLessonsLearned2016,
  title = {Statistical Lessons Learned for Designing Cluster Randomized Pragmatic Clinical Trials from the {{NIH Health Care Systems Collaboratory Biostatistics}} and {{Design Core}}},
  author = {Cook, Andrea J and Delong, Elizabeth and Murray, David M and Vollmer, William M and Heagerty, Patrick J},
  year = 2016,
  journal = {Clinical Trials},
  volume = {13},
  number = {5},
  pages = {504--512},
  publisher = {SAGE Publications},
  issn = {1740-7745},
  doi = {10.1177/1740774516646578},
  urldate = {2026-01-23},
  langid = {english}
}

@article{sandovalDesirabilityOutcomeRanking2024,
  title = {Desirability of Outcome Ranking for Obstetrical Trials: Illustration and Application to the {{ARRIVE}} Trial},
  shorttitle = {Desirability of Outcome Ranking for Obstetrical Trials},
  author = {Sandoval, Grecio J. and Grobman, William A. and Evans, Scott R. and Rice, Madeline M. and Clifton, Rebecca G. and Chauhan, Suneet P. and Costantine, Maged M. and Gibson, Kelly S. and Longo, Monica and Metz, Torri D. and Miller, Emily S. and Parry, Samuel and Reddy, Uma M. and Rouse, Dwight J. and Simhan, Hyagriv N. and Thorp, John M. and Tita, Alan T. N. and Saade, George R.},
  year = 2024,
  month = mar,
  journal = {American Journal of Obstetrics and Gynecology},
  volume = {230},
  number = {3},
  pages = {370.e1-370.e12},
  issn = {0002-9378},
  doi = {10.1016/j.ajog.2023.09.016},
  urldate = {2026-01-23}
}

@book{jungClusterRandomizationTrials2024,
  title = {Cluster Randomization Trials: Statistical Design and Analysis},
  shorttitle = {Cluster {{Randomization Trials}}},
  author = {Jung, Sin-Ho},
  year = 2024,
  month = dec,
  publisher = {{Chapman and Hall/CRC}},
  address = {New York},
  doi = {10.1201/9781003221951},
  isbn = {978-1-003-22195-1}
}

@article{hemmingKeyConsiderationsDesigning2023,
  title = {Key Considerations for Designing, Conducting and Analysing a Cluster Randomized Trial},
  author = {Hemming, Karla and Taljaard, Monica},
  year = 2023,
  month = oct,
  journal = {International Journal of Epidemiology},
  volume = {52},
  number = {5},
  pages = {1648--1658},
  issn = {0300-5771},
  doi = {10.1093/ije/dyad064},
  urldate = {2026-01-26}
}

@book{hunterNotesGraduatelevelCourse,
  title = {Notes for a Graduate-Level Course in Asymptotics for Statisticians},
  year=2014,
  author = {Hunter, David R},
  langid = {english},
  publisher = {The Pennsylvania State University}
}

@article{katheriaApplicationDesirabilityOutcome2024,
  title = {Application of Desirability of Outcome Ranking to the Milking in Non-Vigorous Infants Trial},
  author = {Katheria, Anup C. and {El ghormli}, Laure and Rice, Madeline M. and Dorner, Rebecca A. and Grobman, William A. and Evans, Scott R.},
  year = 2024,
  month = feb,
  journal = {Early Human Development},
  volume = {189},
  pages = {105928},
  issn = {0378-3782},
  doi = {10.1016/j.earlhumdev.2023.105928},
  urldate = {2026-04-23}
}

@article{buyseGeneralizedPairwiseComparisons2010,
  title = {Generalized Pairwise Comparisons of Prioritized Outcomes in the Two-Sample Problem},
  author = {Buyse, Marc},
  year = 2010,
  journal = {Statistics in Medicine},
  volume = {29},
  number = {30},
  pages = {3245--3257},
  issn = {1097-0258},
  doi = {10.1002/sim.3923},
  urldate = {2026-04-24},
  copyright = {Copyright \copyright{} 2010 John Wiley \& Sons, Ltd.},
  langid = {english}
}
